\documentclass[journal,12pt,draftclsnofoot,onecolumn]{IEEEtran}
\usepackage{amsmath} 
\usepackage{graphicx}
\usepackage{epstopdf}
\usepackage{amssymb}
\usepackage{amsthm}
\usepackage{caption}
\usepackage[font={small}]{caption}
\usepackage{cite}
\usepackage{subfigure}
\usepackage{enumerate}
\usepackage{tikz}
\usepackage{subfigure}
\usetikzlibrary{patterns}
\usetikzlibrary{shapes,arrows}
\usepackage{verbatim} 
\usetikzlibrary{positioning}
\usetikzlibrary{shapes.geometric}
\usetikzlibrary{shapes.misc}

\begin{document}
\title{Outage Probability  and Rate for  $\kappa$-$\mu$ Shadowed Fading  in Interference Limited Scenario}
\author{Suman Kumar \hspace*{2.0in} Sheetal Kalyani  \\
\hspace{0in} Dept. of Electrical Engineering, \hspace{1in} Dept. of Electrical Engineering\\
IIT Ropar\hspace{2.2in} IIT Madras\\
{\tt suman@iitrpr.ac.in \hspace{.6in} skalyani@ee.iitm.ac.in}\\
}
\maketitle
\vspace{-0.2in}
\begin{abstract}
The $\kappa$-$\mu$ shadowed fading model is a very general fading model  as it includes both $\kappa$-$\mu$ and $\eta$-$\mu$  as special cases.  In this work, we derive the expression for outage probability when the signal-of-interest (SoI) and interferers both  experience  $\kappa$-$\mu$ shadowed fading  in an interference limited scenario. \textbf{The derived expression is valid for arbitrary SoI parameters, arbitrary $\kappa$ and $\mu$ parameters for all interferers and any value of the parameter $m$ for the interferers excepting the limiting value of $m\rightarrow \infty$. } The expression can be expressed in terms of Pochhammer integral  where the integrands of integral only contains elementary functions.  The outage probability expression is then simplified for various special cases, especially when SoI experiences $\eta$-$\mu$ or $\kappa$-$\mu$ fading.   Further, the  rate expression is derived when the SoI experiences $\kappa$-$\mu$ shadowed fading with  integer values of $\mu$, and interferers experience $\kappa$-$\mu$ shadowed fading with arbitrary parameters. The  rate expression can be expressed in terms of sum of Lauricella's function of the fourth kind.   The utility of our results is demonstrated by  using the derived expression to study and compare FFR and SFR in the presence of $\kappa$-$\mu$ shadowed fading. Extensive simulation results are provided and these further validate our theoretical results.

\end{abstract}
\section{Introduction}
The $\kappa$-$\mu$ fading distribution with two shape parameters, $\kappa$ and $\mu$, and  $\eta$-$\mu$ fading distribution with two shape parameters, $\eta$ and $\mu$  have been proposed to model line-of-sight (LOS) and non line-of-sight (NLOS) propagation effects, respectively in \cite{4231253}. These fading distributions include Nakagami-q (Hoyt), one sided Gaussian, Rayleigh, Nakagami-m, and Rician distribution as special cases.  Recently, in \cite{6594884}, Paris had introduced $\kappa$-$\mu$ shadowed fading  which is a generalization of $\kappa$-$\mu$ fading and Rician shadowed fading\footnote{The Rician shadowed fading has been proposed to model for land mobile satellite (LMS) channels which is a generalization of Rician fading where only the dominant components are subject to shadowing \cite{abdi2003new}. This fading is controlled by two parameters $\kappa$ and $m$, where $m$ is physically related to shadowing.}. In other words, $\kappa$-$\mu$ shadowed fading is a generalization of $\kappa$-$\mu$ fading in which all the dominant components can randomly fluctuate due to shadowing. This fading distribution is controlled by three parameters $\kappa$, $\mu$ and $m$. Note that $\kappa$-$\mu$ shadowed fading has an extra parameter $m$ with respect to the $\kappa$-$\mu$ fading, and an extra parameter $\mu$ with respect to the Rician shadowed fading. 

The $\kappa$-$\mu$ shadowed fading includes the one-side Gaussian, the Rayleigh,  the Nakagami-m, the Rician, the $\kappa$-$\mu$, and Rician shadowed fading distribution as special cases. Very recently, it has been shown in \cite{MorenoPozasLPM15} that the $\eta$-$\mu$ fading distribution is also a special case of $\kappa$-$\mu$ shadowed fading distribution. Hence any analysis for $\kappa$-$\mu$ shadowed fading also holds for $\eta$-$\mu$ and $\kappa$-$\mu$ fading.  The $\kappa$-$\mu$ shadowed fading can be used for modeling and analysis of various wireless communication  systems. For example, it has application to LMS communication \cite{abdi2003new}, underwater communication \cite{ruiz2011ricean}, outdoor device-to-device communication \cite{6953146}, and the access link between an user and its serving relay node \cite{6378481}. Further, since it includes both $\eta$-$\mu$ and $\kappa$-$\mu$ fading \cite{MorenoPozasLPM15}, it is well suited for modeling  LOS and NLOS effects in wireless systems. The $\kappa$-$\mu$ shadowed distribution considers a signal composed of $n$ number of clusters of multipath waves \cite{6594884}. The non-integer $\kappa$-$\mu$ shadowed parameter $\mu$ is the real extension of integer $n$. Note that the  non-integer values of the parameter $\mu$ (i.e., non-integer values of clusters) have been found in practice, and are extensively reported in the literature \cite{996986,4231253} and the references therein. Therefore, it is essential to analyse the system performance in the presence of arbitrary fading.

Considering $\kappa$-$\mu$ and $\eta$-$\mu$  fading distributions, outage  probability has been studied in \cite{6171806, paris2013outage, 6661325, 7130668, 6957529} and references therein.   In particular, the outage probability for $\eta$-$\mu$ faded signal of interest (SoI) and Rayleigh faded interferers has been derived in  \cite{6171806}. The outage probability expression for $\kappa$-$\mu$ or $\eta$-$\mu$  faded signal-of-interest (SoI) with arbitrary parameters is derived in \cite{paris2013outage} where $\eta$-$\mu$ faded co-channel interferences (CCI) with integer value of the fading parameter $\mu$ are considered. In \cite{6661325}, the outage probability is derived  for the following scenarios: $(i)$ when SoI undergoes $\eta$-$\mu$ fading and CCI undergo either $\eta$-$\mu$ or $\kappa$-$\mu$ fading, where the $\kappa$-$\mu$ variates are independent and identically distributed (i.i.d.), and $(ii)$ when SoI undergoes $\kappa$-$\mu$ fading and CCI undergo  $\eta$-$\mu$ fading. However, it has been assumed in \cite{6661325} that the parameter $\mu$ of $\eta$-$\mu$ fading distribution of either the SoI components or CCI components are integers. For $\kappa$-$\mu$ and $\eta$-$\mu$ fading distributions, the rate has been studied in \cite{4400748,5437524,5529757, 6200361, 6056581} and references therein.  However, none of them have  considered co-channel interferers while deriving the rate expression. Expressions for coverage probability and rate are derived in \cite{7130668}, where SoI and interferers experience $\kappa$-$\mu$ and $\eta$-$\mu$ fading, respectively.  An approximate outage probability and channel rate expression are derived in \cite{6502746} and \cite{6957529}  where the user channel experiences $\eta$-$\mu$ fading and $\kappa$-$\mu$ fading, respectively,  and the  interferers  experience $\eta$-$\mu$ fading. To the best of our knowledge, none of the prior works in open literature have derived exact outage probability and rate expressions when both SoI and CCI experience  $\eta$-$\mu$ or both experience $\kappa$-$\mu$ fading  with  arbitrary fading parameters.

Considering $\kappa$-$\mu$ shadowed  fading distributions,  the outage probability and bit error rate expressions  have been derived in \cite{6594884}. The mathematical analysis of the  capacity of the channel assuming $\kappa$-$\mu$ shadowed fading has been presented in \cite{capacity} and \cite{6866213}.  However, none of the previous works on $\kappa$-$\mu$ shadowed fading  have considered co-channel interferers. Approximate expression for outage probability and rate are derived in \cite{7056548} where both  SoI and interferers experience $\kappa$-$\mu$ shadowed fading. To the best of our knowledge, none of the previous works have derived the exact outage probability and rate when  both SoI and interferers experience $\kappa$-$\mu$ shadowed fading.

In this work, we derive the exact outage probability when SoI and interferers both experience $\kappa$-$\mu$ shadowed fading. \textbf{The expression is valid for arbitrary $\kappa$, $\mu$ and $m$ SoI parameters, arbitrary $\kappa$ and $\mu$ parameter for each interferers and for all values of $m$ for the interferers other than $m\rightarrow \infty$ (Note when $m\rightarrow \infty$ for an interferer, the interferer experience $\kappa$-$\mu$ fading).}   The derived expression is given in terms of Pochhammer integral  where the integrands of integral only contains elementary functions. Also, the  expression is given  in terms of sum of Lauricella's function of the fourth kind, which can also be easily evaluated numerically.  Then we simplify the expressions for the following cases: $(a)$ When SoI experience $\kappa$-$\mu$ fading, $(b)$ When SoI and interferers both experience $\eta$-$\mu$  fading with arbitrary parameters, $(c)$ When SoI experience Hoyt fading  and interferers experience $\eta$-$\mu$  fading with arbitrary parameters.  Interestingly, the outage probability expressions when SoI experience Hoyt fading is given in terms of single Lauricella's function of the fourth kind. \textbf{ The outage probability expression when SoI and interferers both experience $\kappa$-$\mu$ shadowed fading can not be simplified for the case when the interferers  experience  $\kappa$-$\mu$ fading and this is limitation of this expression. However, it has been mentioned in \cite{7870713} that the $\kappa$-$\mu$ shadowed fading model can be used to approximate the $\kappa$-$\mu$ distribution with arbitrary precision, by simply choosing a sufficiently large value of $m$. In other words, one can evaluate the outage probability when the interferers experience $\kappa$-$\mu$ fading by simply choosing a sufficiently large value of $m$ in the derived expression and evaluating it.}

Further, the exact rate expression is derived when SoI experience $\kappa$-$\mu$ shadowed fading with integer $\mu$  and interferers experience $\kappa$-$\mu$ shadowed fading with arbitrary parameters. \textbf{Again, the rate expression is valid for arbitrary $\kappa$ and $m$ SoI parameters, arbitrary $\kappa$ and $\mu$ parameter for each interferers and for all values of $m$ for the interferers other than $m\rightarrow \infty$.}  The rate expression is given when SoI experience $\eta$-$\mu$  fading with integer $\mu$  and interferers experience $\eta$-$\mu$  fading with arbitrary parameters.  Again, the rate expression can be expressed in terms of sum of Lauricella's function of the fourth kind. Simulation results are  provided and these match with the derived results. The impact of the $\kappa$-$\mu$ shadowed fading parameters namely $\kappa$, $\mu$ and $m$ on both outage probability and rate is also studied. Finally, we show the utility of our results by using them for comparison of fractional frequency reuse (FFR) and soft frequency reuse (SFR) in the presence of generalized fading.

FFR and SFR are the two  popular inter-cell interference coordination (ICIC) schemes. Comparatively, FFR is a simpler scheme but, SFR may be more bandwidth efficient. The performance of FFR and SFR are extensively studied and compared in \cite{6047548, 6898829, 6952716, 7015633,7137606} and references therein. In particular,  both  schemes are compared in \cite{6898829} under fully loaded and partial loaded system while considering Nakagami-m fading. It has been shown in \cite{6898829} that in a partial loaded system, SFR outperforms FFR, whereas in a fully loaded system opposite is true. Considering Rayleigh fading, both  schemes are compared in \cite{6047548,7015633,7137606}. To the best of our knowledge, none of the previous work has compared the performance of FFR and SFR schemes in the presence of generalized fading, such as $\kappa$-$\mu$, $\eta$-$\mu$ and $\kappa$-$\mu$ shadowed fading. In this work, we compare both schemes in the presence of  generalized fading, i.e., $\kappa$-$\mu$ shadowed fading as it includes both $\kappa$-$\mu$ and $\eta$-$\mu$ fading. We show that in a fully loaded system, FFR outperforms SFR in the presence of generalized fading. 

In this work, $O_p$ denotes the outage probability,  $\Phi_2^{(N)}(.)$ is confluent multivariate hypergeometric function and $F_D^{(N)}(.) $  denotes Lauricella function of fourth kind \cite{exton1976multiple}. ${}_{(1)}^{(1)}E_D^{(N)}(.)$ is closely related to $F_D^{(N)}(.)$ \cite{exton1976multiple} and $C_N^{(k)}(.)$ is a generalization of the Horn function\cite{exton1976multiple}. 
\section{System Model}
We consider the downlink homogeneous macrocell network with hexagonal structure with radius $R$ as shown in Fig. \ref{fig:hexagonal}. The Signal-to-Interference-Ratio (SIR) is expressed as
\begin{equation}
\text{SIR}=\frac{g'r^{-\alpha}}{\sum\limits_{i=1}^{N}h'_id_i^{-\alpha}}=\frac{g}{I}, \text{ where } I=\sum\limits_{i=1}^{N}h_i. \label{macro_sir}
\end{equation}
Here $g'$ and $h'_i$ are the small scale fading experienced by SoI and the $i^{th}$ interferer, respectively. A standard path loss model $r^{-\alpha}$ is considered, where $\alpha\geq 2$ is the path loss exponent.  The distance from user to serving base station (BS) and $i^{th}$ interferer are denoted by $r$ and $d_i$, respectively. An interference limited network is considered and hence noise is neglected\footnote{Note that considering  noise will lead in extremely complicated expression for outage probability which could be very hard to simplify for the arbitrary parameter case. Hence, in this work, our focus  is on SIR instead of SINR since it leads to both novel and mathematically tractable expression. Further, we have conducted extensive simulations and they indicate that ignoring noise does not significantly affect the numerical results. }.   The number of interferers is denoted  by $N$. The fading gain $g'$ is $\kappa$-$\mu$ shadowed distributed with mean $\bar{\gamma}'$, i.e., $\bar{\gamma}'=E[g']$  and shape parameters $\kappa$, $\mu$, and $m$. Note that the $\kappa$-$\mu$ shadowed fading considers a signal composed of clusters of multipath waves. Within each cluster, the phases of the scattered waves are random and have similar delay times. However, the inter-cluster delay-time spreads are assumed relatively large. Moreover, within each cluster a dominant component exists, which can randomly fluctuate because of shadowing. On the other hand, in $\kappa$-$\mu$ fading, a deterministic dominant component exists within each cluster. Note $\mu$ is the real extension of the number of clusters, $\kappa$ is the ratio between the total power of the dominant components and the total power of the scattered waves. The shape parameter $m$ is related to shadowing component. Also,  $h'_i$s are $\kappa$-$\mu$ shadowed distributed  with mean $\bar{\gamma}'_i$, i.e., $\bar{\gamma}'_i=E[h'_i]$  and shape parameters $\kappa_i$, $\mu_i$, $m_i$. Here a random variable $x$ with mean $\bar{x}$ and shaping parameters $\kappa$, $\mu$ and $m$ is  symbolically expressed as $x\sim S_{\kappa\mu}(\kappa,\mu,m,\bar{x})$. Thus, $g'\sim S_{\kappa\mu}(\kappa,\mu,m,\bar{\gamma}')$ and $h'_i\sim S_{\kappa\mu}(\kappa_i,\mu_i,m_i,\bar{\gamma}'_i)$. Note that $g=g'r^{-\alpha}$ and $h_i=h'_id_i^{-\alpha}$ and hence  $g\sim S_{\kappa\mu}(\kappa,\mu,m,\bar{\gamma})$ and  $h_i\sim S_{\kappa\mu}(\kappa_i,\mu_i,m_i,\bar{\gamma}_i)$ where $\bar{\gamma}=\bar{\gamma}'r^{-\alpha}$, $\bar{\gamma}_i=\bar{\gamma}'_id_i^{-\alpha}$. Although, we have considered the homogeneous macrocell network with hexagonal structure, the derivation of the paper is valid for any scenario where the distance from the serving BS and the distances from the interferers are given.
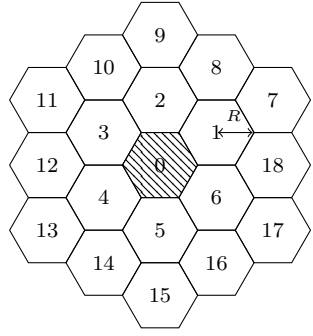
\begin{figure}[ht]
\centering
\begin{tikzpicture}
\node[pattern=north west lines, regular polygon, regular polygon sides=6,minimum width= 1  cm, draw] at (5,5) {};
\node at (5,5){\scriptsize $0$};
\node[regular polygon, regular polygon sides=6,minimum width=1 cm, draw] at (5+1.5*0.5,5-0.866*0.5) {};
\node at (5+1.5*0.5,5-0.866*0.5){\scriptsize $6$};
\node[ regular polygon, regular polygon sides=6,minimum width=1 cm, draw] at (5+1.5*0.5,5+0.866*0.5) {};
\node at (5+1.5*0.5,5+0.866*0.5){\scriptsize $1$};
\node[ regular polygon, regular polygon sides=6,minimum width=1 cm, draw] at (5-1.5*0.5,5+0.866*0.5) {};
\node at (5-1.5*0.5,5+0.866*0.5){\scriptsize $3$};
\node[ regular polygon, regular polygon sides=6,minimum width=1 cm, draw] at (5-1.5*0.5,5-0.866*0.5) {};
\node at (5-1.5*0.5,5-0.866*0.5){\scriptsize $4$};
\node[ regular polygon, regular polygon sides=6,minimum width=1 cm, draw] at (5,5-0.866) {};
\node at (5,5-0.866){\scriptsize $5$};
\node[ regular polygon, regular polygon sides=6,minimum width=1 cm, draw] at (5,5+0.866) {};
\node at (5,5+0.866){\scriptsize $2$};
\node[ regular polygon, regular polygon sides=6,minimum width=1 cm, draw] at (5,5+0.866*2) {};
\node at (5,5+0.866*2){\scriptsize $9$};
\node[ regular polygon, regular polygon sides=6,minimum width=1 cm, draw] at (5,5-0.866*2) {};
\node at (5,5-0.866*2){\scriptsize $15$};
\node[ regular polygon, regular polygon sides=6,minimum width=1 cm, draw] at (5+3*0.5,5+0.866*1) {};
\node at (5+3*0.5,5+0.866*1){\scriptsize $7$};
\node[ regular polygon, regular polygon sides=6,minimum width=1 cm, draw] at (5+3*0.5,5-0.866*1) {};
\node at (5+3*0.5,5-0.866*1){\scriptsize $17$};
\node[ regular polygon, regular polygon sides=6,minimum width=1 cm, draw] at (5-3*0.5,5-0.866*1) {};
\node at (5-3*0.5,5-0.866*1){\scriptsize $13$};
\node[ regular polygon, regular polygon sides=6,minimum width=1 cm, draw] at (5-3*0.5,5+0.866*1) {};
\node at (5-3*0.5,5+0.866*1){\scriptsize $11$};
\node[ regular polygon, regular polygon sides=6,minimum width=1 cm, draw] at (5+1.5,5) {};
\node at (5+1.5,5){\scriptsize $18$};
\node[ regular polygon, regular polygon sides=6,minimum width=1 cm, draw] at (5-1.5,5) {};
\node at (5-1.5,5){\scriptsize $12$};
\node[ regular polygon, regular polygon sides=6,minimum width=1 cm, draw] at (5+1.5*0.5,5-0.866*1.5) {};
\node at(5+1.5*0.5,5-0.866*1.5){\scriptsize $16$};
\node[ regular polygon, regular polygon sides=6,minimum width=1 cm, draw] at (5-1.5*0.5,5+0.866*1.5) {};
\node at(5-1.5*0.5,5+0.866*1.5){\scriptsize $10$};
\node[ regular polygon, regular polygon sides=6,minimum width=1 cm, draw] at (5-1.5*0.5,5-0.866*1.5) {};
\node at(5-1.5*0.5,5-0.866*1.5){\scriptsize $14$};
\node[ regular polygon, regular polygon sides=6,minimum width=1 cm, draw] at (5+1.5*0.5,5+0.866*1.5) {};
\node at(5+1.5*0.5,5+0.866*1.5){\scriptsize $8$};
\draw[<->] (5+1.5*0.5+0.02,5+0.866*0.5)--(5+1.5*0.5-0.02+0.5,5+0.866*0.5) node[pos=0.45,sloped,above] {\tiny$R$};
\end{tikzpicture}
            \caption{ Macrocell network with hexagonal tessellation having radius $R$.}
           
             \label{fig:hexagonal}
        \end{figure}

The probability density function (pdf) of $\kappa$-$\mu$ shadowed distribution is given by \cite{6594884}
\begin{equation}
\textstyle
f_X(x)=\frac{\mu^\mu m^m (1+\kappa)^\mu x^{\mu-1}}{\Gamma(\mu)(\bar{\gamma})^\mu(\mu\kappa+m)^m}e^{-\frac{\mu(1+\kappa)x}{\bar{\gamma}}}  {}_1 F_1 \left(m,\mu,\frac{\mu^2\kappa(1+\kappa)}{\mu\kappa+m}\frac{x}{\bar{\gamma}}\right), \label{kappa_mu_shadowed}
\end{equation}
where ${}_1 F_1$ is confluent hypergeometric function \cite{exton1976multiple}. Recall that $\kappa$-$\mu$ shadowed is the generalization of the Rician shadowed fading and $\kappa$-$\mu$ fading. Rician shadowed fading considers only one cluster and within that cluster a dominant component exists which can randomly fluctuate because of shadowing. Therefore, by simply putting $\mu=1$ in \eqref{kappa_mu_shadowed}, one obtains the pdf of Rician shadowed fading given in \cite[Eq. (6)]{abdi2003new}. Similarly, $\kappa$-$\mu$ fading considers each cluster having deterministic dominant component. Therefore, by putting $m \rightarrow \infty$ in \eqref{kappa_mu_shadowed}, which makes the dominant component to be deterministic, one obtains $\kappa$-$\mu$ fading given in\cite[Eq.(2)]{4231253}.   Assuming $\theta=\frac{\bar{\gamma}}{\mu(1+\kappa)}$ and $\lambda=\frac{(\mu\kappa+m)\bar{\gamma}}{\mu(1+\kappa)m}$, the  above pdf can be rewritten as
\begin{equation}
\textstyle
f_X(x)=\frac{x^{\mu-1}} {\theta^{\mu-m}\lambda^m\Gamma(\mu)}e^{-\frac{x}{\theta}}
  {}_1 F_1 \left(m,\mu,\frac{x}{\theta}-\frac{x}{\lambda}\right)
\end{equation}
 The cumulative distribution function (cdf) of $\kappa$-$\mu$ shadowed fading is given by \cite{6594884}
\begin{equation}
\textstyle
F_X(x)=\frac{x^{\mu} } {\theta^{\mu-m}\lambda^m\Gamma(\mu+1)} \, \Phi_2 \left(\mu-m,m,\mu+1,-\frac{x}{\theta},-\frac{x}{\lambda}\right).
\end{equation}
where $\Phi_2(.)$ denotes bivariate confluent hypergeometric function \cite{exton1976multiple}. The pdf of the $\eta$-$\mu$ RV  is given by \cite{4231253}.
\begin{equation}
\textstyle
f_{Y}(y)=\frac{2\sqrt{\pi}\bar{\mu}^{\bar{\mu}+ \frac{1}{2}}h^{\bar{\mu}} y^{\bar{\mu}- \frac{1}{2}}}{\Gamma(\bar{\mu})H^{\bar{\mu}- \frac{1}{2}}}\exp(-2\bar{\mu} hy)I_{\bar{\mu}- \frac{1}{2}}(2\bar{\mu} Hy)\label{eta_mu}
\end{equation}
where $\bar{\mu}=\frac{E^2\{Y\}}{2\text{var}\{Y\}}[1+(\frac{H}{h})^2]$ (with expectation and variance being denoted by $E\{.\}$ and $\text{var}\{.\}$, respectively.), $\Gamma(.)$ denotes the gamma function, and $I_{\bar{\mu}}(.)$ denotes the modified Bessel function of the first kind of the order $\bar{\mu}$. Parameters $H$ and $h$ can both be defined  in two different ways corresponding to two different fading formats depending on the physical meaning of the parameter $\eta$. In format $(i)$, $0<\eta< \infty $ is the power ratio of the in-phase and quadrature 	component of the fading signal in each multipath cluster, and $H$ and $h$ are given by $H=\frac{\eta^{-1}-\eta}{4}, \text{ and  } h=\frac{2+\eta^{-1}+\eta}{4}.$ 
Format $(ii)$ can be obtained from format $(i)$ \cite{4231253}. In order to show the utility of our results, we have used them to compare FFR and SFR performance metrics. Therefore, in next paragraph, we briefly discuss both FFR and SFR.

FFR and SFR both divide the users as cell-centre users and cell-edge users based on a predefined SIR threshold $S_{t}$  as shown in Fig. \ref{fig:sfr}. Users with SIR $\geq S_{t}$ are classified as cell-centre users, while remaining users are assumed to be cell-edge users.  FFR uses frequency reuse $1$ at the cell-centre and frequency reuse $\frac{1}{3}$ at the cell-edge. Whereas, SFR uses frequency reuse $\frac{1}{3}$ at the cell-edge and uses the cell-edge frequency of other cells at the cell-centre. Also, in SFR the transmit power levels for cell-edge users is $\beta$ times higher than the cell-centre users.

       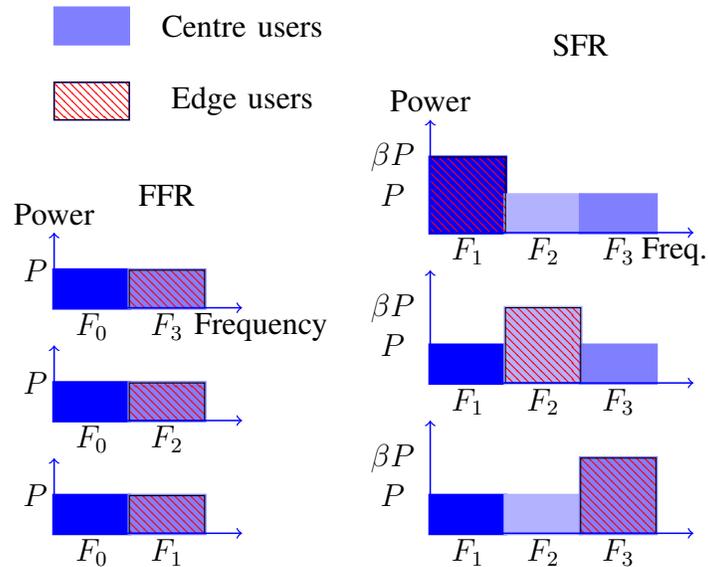
\begin{figure}[ht]
\centering
\begin{tikzpicture}
\node at (1.5,1){$\beta  P $};
\node at (1.5,0.5){$P$};
\node at (2.5,-0.25){$F_1$};
\node at (3.5,-0.25){$F_2$};
\node at (4.5,-0.25){$F_3$};
\draw[fill] [blue!100,line width=.05cm](2,0) rectangle (3,0.5);
\draw[fill] [blue!30,line width=.05cm](3,0) rectangle (4,0.5);
\draw[fill] [blue!50,line width=.05cm](4,0) rectangle (5,1);
\draw[pattern=north west lines, pattern color=red](4,0) rectangle (5,1);
\draw [blue,line width=.02  cm] [->](2,0)--(2,1.5);
\node at (1.5,3){$\beta P$};
\node at (1.5,2.5){$P$};
\node at (2.5,1.75){$F_1$};
\node at (3.5,1.75){$F_2$};
\node at (4.5,1.75){$F_3$};
\draw[fill] [blue!100,line width=.05cm](2,2) rectangle (3,2.5);
\draw[fill] [blue!30,line width=.05cm](3,2) rectangle (4,3);
\draw[fill] [blue!50,line width=.05cm](4,2) rectangle (5,2.5);
\draw[pattern=north west lines, pattern color=red](3,2) rectangle (4,3);
\draw [blue,line width=.02  cm] [->](2,2)--(2,3.5);
\node at (1.5,5){$\beta P$};
\node at (1.5,4.5){$P$};
\node at (2.5,3.75){$F_1$};
\node at (3.5,3.75){$F_2$};
\node at (4.5,3.75){$F_3$};
\draw[fill] [blue!100,line width=.05cm](2,4) rectangle (3,5);
\draw[fill] [blue!30,line width=.05cm](3,4) rectangle (4,4.5);
\draw[fill] [blue!50,line width=.05cm](4,4) rectangle (5,4.5);
\draw[pattern=north west lines, pattern color=red](2,4) rectangle (3,5);
\draw [blue,line width=.02  cm] [->](2,4)--(2,5.5);
\node at (2,5.75){Power};
\draw [blue,line width=.02  cm] [->](2,0)--(5.5,0);
\draw [blue,line width=.02  cm] [->](2,2)--(5.5,2);
\draw [blue,line width=.02  cm] [->](2 ,4)--(5.5,4);
\node at (5.25,3.75){Freq.};
\node at (4,6.5){SFR};
\draw[fill] [blue!50,line width=.02cm](-3,6.5) rectangle (-2,7);
\draw [blue,line width=.02cm](-3,5.5) rectangle (-2,6);
\draw[pattern=north west lines, pattern color=red](-3,5.5) rectangle (-2,6);
\node at (-0.5,6.75){Centre users};
\node at (-0.5,5.75){Edge users};
\node at (1.75-5,0.5){$P$};
\node at (2.5-5,-0.25){$F_0$};
\node at (3.5-5,-0.25){$F_1$};
\draw[fill] [blue!100,line width=.05cm](2-5,0) rectangle (3-5,0.5);
\draw[fill] [blue!50,line width=.05cm](3-5,0) rectangle (4-5,0.5);
\draw[pattern=north west lines, pattern color=red](3-5,0) rectangle (4-5,0.5);
\draw [blue,line width=.02  cm] [->](2-5,0)--(2-5,1);
\node at (1.75-5,2){$P$};
\draw[fill] [blue!100,line width=.05cm](2-5,1.5) rectangle (3-5,2);
\draw[fill] [blue!50,line width=.05cm](3-5,1.5) rectangle (4-5,2);
\draw[pattern=north west lines, pattern color=red](3-5,1.5) rectangle (4-5,2);
\draw [blue,line width=.02  cm] [->](2-5,1.5)--(2-5,2.5);
\node at (1.75-5,3.5){$P$};
\draw[fill] [blue!100,line width=.05cm](2-5,3) rectangle (3-5,3.5);
\draw[fill] [blue!50,line width=.05cm](3-5,3) rectangle (4-5,3.5);
\draw[pattern=north west lines, pattern color=red](3-5,3) rectangle (4-5,3.5);
\draw [blue,line width=.02  cm] [->](2-5,3)--(2-5,4);
\node at (2-5,4.25){Power};
\draw [blue,line width=.02  cm] [->](2-5,0)--(4.5-5,0);
\draw [blue,line width=.02  cm] [->](2-5,1.5)--(4.5-5,1.5);
\draw [blue,line width=.02  cm] [->](2-5,3)--(4.5-5,3);
\node at (5-5.25,2.75){Frequency};
\node at (4-5.5,4.5){FFR};
\node at (2.5-5,1.25){$F_0$};
\node at (3.5-5,1.25){$F_2$};
\node at (2.5-5,2.75){$F_0$};
\node at (3.5-5,2.75){$F_3$};
\end{tikzpicture}
        \caption{Frequency and power allocation in FFR and SFR for three neighbouring cells.}
        \label{fig:sfr}
        \end{figure} 

\section{ Outage Probability and Rate}

In this section, we first derive the general expression of outage  probability when both  user and interferers experience $\kappa$-$\mu$ shadowed fading. Then, we simplify the general expression for various special cases.  The outage probability  is the probability  that a user cannot achieve a target SIR $T$ and it can be written as $P(\text{SIR}<T)$.
\newtheorem{theorem}{Theorem}
\begin{theorem}
The outage probability at a given distance when both user and interferers experience $\kappa$-$\mu$ shadowed fading  with arbitrary parameters is given by
\begin{equation*}
\textstyle
O_p=1-K\bigg(\frac{T\theta_1}{\theta +T\theta_1}\bigg)^{\sum\limits_{i=1}^{N}\mu_i+\mu}
{}_{(1)}^{(1)}E_D^{(2N+1)}\bigg[\sum\limits_{i=1}^{N}\mu_i+\mu, m,\mu_1-m_1,\cdots,  
 \mu_N-m_N,m_1\cdots,
\end{equation*}
\begin{equation}
\textstyle
m_{N-1},1;\mu,1+\sum\limits_{i=1}^{N}\mu_i;\frac{(\lambda-\theta)T\theta_1}{(T\theta_1+\theta)\lambda}, \frac{\theta}{\theta +T\theta_1}, \frac{\theta\theta_2-\theta\theta_1}{\theta_{2}(\theta+T\theta_1)},\cdots,\frac{\theta\lambda_{N}-\theta\theta_1}{\lambda_{N}(\theta+T\theta_1)}\bigg] \label{final1}.
\end{equation}
\end{theorem}
\begin{proof}
See Appendix \ref{outage_prob} for the proof.
\end{proof}
 Note that the outage probability expression is given in terms of the function $E_D^{(2N+1)}(.)$. The function ${}_{(1)}^{(1)}E_D^{(N)}(.)$ is closely related to Lauricella function of fourth kind $F_D^{(N)}(.) $\cite{exton1976multiple} and the series expression for $E_D^{(2N+1)}(.)$ is given in \eqref{eq:lauricella3}.  However, there is no code (either Matlab or Mathematica) available for ${}_{(1)}^{(1)}E_D^{(N)} (.)$. Also there is no single integral expression is available for ${}_{(1)}^{(1)}E_D^{(N)} (.)$. Hence in the Appendix \ref{evaluation}, we have simplified the outage probability further such that it can be evaluated easily. The outage probability  in terms of $F_D(.)$ is given by (See Appendix \ref{evaluation} for the proof) 
\begin{equation*}
\textstyle
O_p=K_2\sum\limits_{p=0}^{\infty} \frac{(m)_p\Big(1-\frac{\theta}{\lambda}\Big)^p\Gamma(\sum\limits_{i=1}^{N}\mu_i+p+\mu)}{(\mu)_p p!} F_D^{(2N)}\bigg(1-p-\mu,\mu_1-m_1,\cdots, \mu_N-m_N,m_1\cdots,m_N;
\end{equation*} 
\begin{equation}
\textstyle
 1+\sum\limits_{i=1}^{N}\mu_i; \frac{\theta}{\theta+T\theta_1},\cdots,\frac{\theta}{\theta+T\theta_N}, \frac{\theta}{\theta+T\lambda_1},\cdots, \frac{\theta}{\theta+T\lambda_N}\bigg) \label{kappa_mu4},
\end{equation}
where $K_2=\frac{\left(\prod\limits_{i=1}^{N}\left(\frac{\theta}{\theta+T\theta_i}\right)^{\mu_i-m_i}\left(\frac{\theta}{\theta+T\lambda_i}\right)^{m_i}\right)}{\Gamma\left(1+\sum\limits_{i=1}^{N}\mu_i\right)}\frac{\theta^ m }{\Gamma(\mu)(\lambda)^m}$. 
Note that the above expression is in terms of infinite series.  It is analytically shown in Appendix \ref{truncation}   that the  infinite series can be truncated to a series with a finite number of terms with truncation error being lower than a user specified $\epsilon$ ($\epsilon$ is typically of the order of  $10^{-5}$ or lower). In other words, instead of  evaluating the infinite series it is sufficient to use the first $P$ terms as given in \eqref{kappa_mu5}, where $P$ can be chosen according to the accuracy of results required. It is also shown in Appendix \ref{truncation} that as $\kappa$ or $\mu$ increases the number of terms required, i.e., $P$ increases.  Denoting by $O_{p,P}$, the outage probability when only first $P$ terms are considered is given by 
\begin{equation*}
\textstyle
O_{p,P}=K_2\sum\limits_{p=0}^{P} \frac{(m)_p\Big(1-\frac{\theta}{\lambda}\Big)^p\Gamma(\sum\limits_{i=1}^{N}\mu_i+p+\mu)}{(\mu)_p p!} F_D^{(2N)}\bigg(1-p-\mu,\mu_1-m_1,\cdots,
\end{equation*} 
\begin{equation}
\textstyle
 \mu_N-m_N,m_1\cdots,m_N;1+\sum\limits_{i=1}^{N}\mu_i;\frac{\theta}{\theta+T\theta_1},\cdots,\frac{\theta}{\theta+T\theta_N}, \frac{\theta}{\theta+T\lambda_1},\cdots, \frac{\theta}{\theta+T\lambda_N}\bigg) \label{kappa_mu5}.
\end{equation}
 The expression is valid for general path loss exponent, and hence the derived expression can be used for both NLOS and LOS path loss defined in \cite{7416971}.
Fig. \ref{cov_kappa_terms} and Fig. \ref{cov_mu_terms} plots the outage probability with respect to number of terms $P$ for different value of $\kappa$ and $\mu$, respectively. It can be seen that the outage probability converges fast even for higher value of $\kappa$ and $\mu$. As shown in Appendix \ref{truncation}  it can be also observed that as $\kappa$ or $\mu$ increases the number of terms required, i.e., $P$ increases.  Note that Fig. \ref{cov_kappa_terms} and Fig. \ref{cov_mu_terms} is plotted using \eqref{kappa_mu5}.

 \begin{figure}[ht]
  \centering
 \includegraphics[scale=0.45]{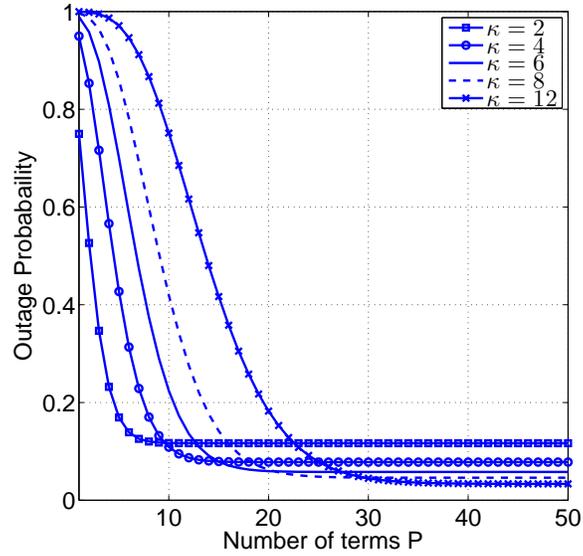}
 \caption{Variation in outage probability with respect to number of terms $P$ for different values of $\kappa$. Here  $r=650$m $\bar{\gamma}=\bar{\gamma}_i=1 \forall i$, $T=0$ dB, $\mu=1.8$  $m=8$,  $\kappa_i=1.5, \mu_i=1.5 \forall i\in \{1,\cdots,6\}$, $\kappa_i=1, \mu_i=1 \forall i\in \{7,\cdots,12\}$, $\kappa_i=0.5, \mu_i=0.5 \forall i\in \{13,\cdots,18\}$ $m_i=10 \forall i$, and $\alpha=4$.}
 \label{cov_kappa_terms}
 \end{figure}

 \begin{figure}[ht]
  \centering
 \includegraphics[scale=0.45]{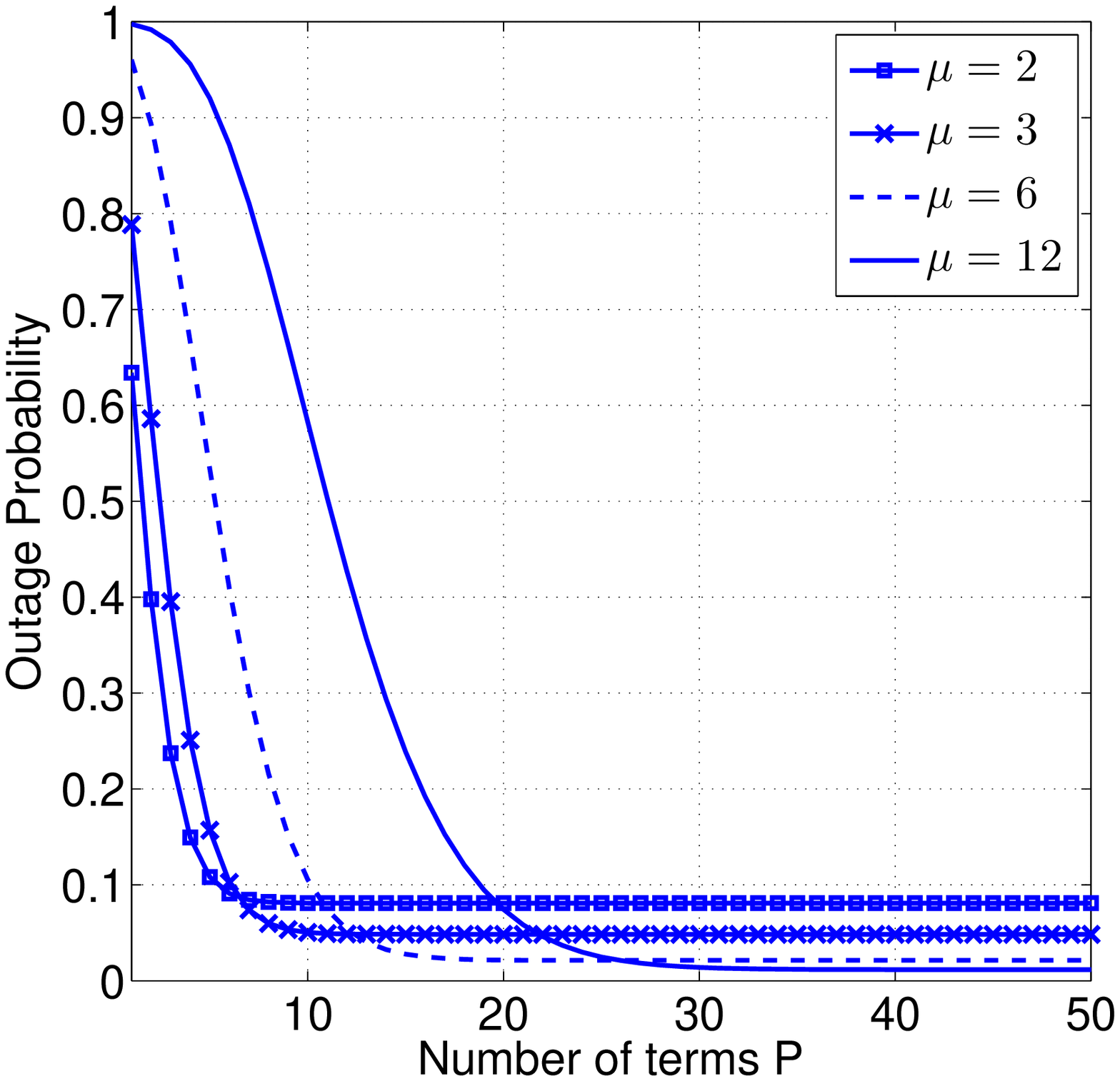}
 \caption{Variation in outage probability with respect to number of terms $P$ for different values of $\mu$. Here $r=650$m $\bar{\gamma}=\bar{\gamma}_i=1 \forall i$, $T=0$ dB, $\kappa=2.5$  $m=8$, $\kappa_i=1.5, \mu_i=1.5 \forall i\in \{1,\cdots,6\}$, $\kappa_i=1, \mu_i=1 \forall i\in \{7,\cdots,12\}$, $\kappa_i=0.5, \mu_i=0.5 \forall i\in \{13,\cdots,18\}$, $m_i=10 \forall i$, and $\alpha=4$.}
 \label{cov_mu_terms}
 \end{figure}

We now discuss how the general expression given in here can be simplified for various cases using  properties of special functions, when the interferers always experience $\kappa$-$\mu$ shadowed fading with arbitrary\footnote{By arbitrary we mean arbitrary $\kappa$ values, arbitrary $\mu$ values and arbitrary $m$ values excepting $m\rightarrow \infty$ for all interferers } parameters

\subsection{Outage probability expression when SoI experiences $\kappa$-$\mu$  fading}
In this subsection, we obtain the outage probability when desired channel experience $\kappa$-$\mu$ fading ($m\rightarrow \infty$) and interferers experience $\kappa$-$\mu$ shadowed fading with arbitrary parameters. In order to obtain the outage probability we need to evaluate $\lim\limits_{m\rightarrow \infty} O_p$, where $O_p$ is given in  \eqref{kappa_mu4}. Using the fact that $\lim\limits_{m\rightarrow \infty} \left(\frac{m}{a+m} \right)^m=\exp(-a)$ and $\lim\limits_{m\rightarrow \infty}(m)_p  \left(\frac{a}{a+m} \right)^p =a^p$, one obtains the outage probability expression when desired channel experience $\kappa$-$\mu$ fading to be
\begin{equation*}
\textstyle
O_p=K_2\sum\limits_{p=0}^{\infty} \frac{(\mu\kappa)^p \Gamma(\sum\limits_{i=1}^{N}\mu_i+p+\mu)}{(\mu)_p p!} F_D^{(2N)}\bigg(1-p-\mu,\mu_1-m_1,\cdots,
\end{equation*} 
\begin{equation}
\textstyle
 \mu_N-m_N,m_1\cdots, m_N;1+\sum\limits_{i=1}^{N}\mu_i;\frac{\theta}{\theta+T\theta_1},\cdots,\frac{\theta}{\theta+T\theta_N}, \frac{\theta}{\theta+T\lambda_1},\cdots, \frac{\theta}{\theta+T\lambda_N}\bigg) \text{d}t\label{kappa_mu40}
\end{equation}
where $K_2=\frac{\left(\prod\limits_{i=1}^{N}\left(\frac{\theta}{\theta+T\theta_i}\right)^{\mu_i-m_i}\left(\frac{\theta}{\theta+T\lambda_i}\right)^{m_i}\right)}{\Gamma\left(1+\sum\limits_{i=1}^{N}\mu_i\right)}\frac{\exp(-\mu\kappa) }{\Gamma(\mu)}$. 
Note that the outage probability expression when SoI experience $\kappa$-$\mu$ fading and interferers experience $\eta$-$\mu$ fading is given in \cite[Eq. (21)]{7130668}. The expression given in \eqref{kappa_mu40} can be reduced to the expression given in \cite[Eq. (21)]{7130668}, by putting $\mu_i=2\bar{\mu}_i, \kappa_i=\frac{1-\eta_i}{2\eta_i}, m=\bar{\mu}_i \forall i$. 
\subsection{Outage probability expression when both SoI and interference  experience $\eta$-$\mu$  fading}
In this subsection, we obtain the outage probability when desired channel and interferers both experience $\eta$-$\mu$ fading. Recently, it has been shown in \cite{MorenoPozasLPM15} that $\eta$-$\mu$ is a special case of $\kappa$-$\mu$ shadowed fading and it can be obtained by putting $\mu=2\bar{\mu}, \kappa=\frac{1-\eta}{2\eta}, m=\bar{\mu}$ (for format $1$). Using the above transformation, one obtains the outage probability expression when desired channel and interferers both experience $\eta$-$\mu$ fading and it is given by
\begin{eqnarray}
\textstyle
O_p=1-\bar{K}\bigg(\frac{T\beta_1}{a_2 +T\beta_1}\bigg)^{\sum\limits_{i=1}^{2N}\bar{\mu}_i+2\bar{\mu}}
{}_{(1)}^{(1)}E_D^{(2N+1)}\bigg[\sum\limits_{i=1}^{2N}\bar{\mu}_i+2\bar{\mu}, \bar{\mu},\bar{\mu}_1, \bar{\mu}_1, \cdots,  
 \bar{\mu}_N, \bar{\mu}_N,1; \nonumber \\
\textstyle 2\bar{\mu},1+\sum\limits_{i=1}^{2N}\bar{\mu}_i;\frac{(a_1-a_2)T\beta_1}{(T\beta_1+a_2)a_1}, \frac{a_2}{a_2 +T\beta_1}, \frac{a_2\beta_2-a_2\beta_1}{\beta_{2}(a_2+T\beta_1)},\cdots,\frac{a_2\beta_{2N}-a_2\beta_1}{\beta_{2N}(a_2+T\beta_1)}\bigg] \label{final1a}.
\end{eqnarray}
where  $\bar{K}=\frac{\Gamma(\sum\limits_{i=1}^{2N}\bar{\mu}_i+2\bar{\mu})\left(\prod\limits_{i=1}^{N}\frac{1}{\beta_{2i}^{\mu_{i}}\beta_{2i-1}^{\mu_i}}\right)}{\Gamma\left(1+\sum\limits_{i=1}^{2N}\bar{\mu}_i\right){T}^{\sum\limits_{i=1}^{2N}\bar{\mu}_i}}\frac{a_2^{\sum\limits_{i=1}^{2N}\bar{\mu}_i+\bar{\mu}}} {a_1^{\bar{\mu}}\Gamma(\bar{\mu})}$
$a_1=\frac{1}{2\bar{\mu}(h+H)}$, $a_2=\frac{1}{2\bar{\mu}(h-H)}$, $\beta_{2i-1}=\frac{1}{2\bar{\mu}_i(h_i+H_i)}$, $\beta_{2i}=\frac{1}{2\bar{\mu}_i(h_i-H_i)}$ $H_i=\frac{\eta_i^{-1}-\eta_i}{4}, \text{ and  } h=\frac{2+\eta_i^{-1}+\eta_i}{4}.$

\subsection{Outage probability when the user channel undergoes Hoyt (Nakagami-q) fading}
When for the user channel $\bar{\mu}=\frac{1}{2}$ (note  Hoyt (or Nakagami-q) distribution is a special case of the $\eta$-$\mu$ distribution  with $\bar{\mu}=\frac{1}{2}$ \cite{4231253}), then the outage probability given in \eqref{final1a}  reduces to
\begin{eqnarray*}
\textstyle
O_p=1-\bar{K}\bigg(\frac{T\beta_1}{a_2 +T\beta_1}\bigg)^{\sum\limits_{i=1}^{2N}\bar{\mu}_i+1}
{}_{(1)}^{(1)}E_D^{(2N+1)}\bigg[\sum\limits_{i=1}^{2N}\bar{\mu}_i+1, \frac{1}{2},
\end{eqnarray*}
\begin{equation}
\textstyle
\bar{\mu}_1, \bar{\mu}_1, \cdots,  
\bar{\mu}_N, \bar{\mu}_N,1;  1,1+\sum\limits_{i=1}^{2N}\bar{\mu}_i; \frac{(a_1-a_2)T\beta_1}{(T\beta_1+a_2)a_1},
 \frac{a_2}{a_2 +T\beta_1},  \frac{a_2\beta_2-a_2\beta_1}{\beta_{2}(a_2+T\beta_1)},\cdots,\frac{a_2\beta_{2N}-a_2\beta_1}{\beta_{2N}(a_2+T\beta_1)}\bigg] \label{final1a1}.
\end{equation}
$\bar{K}=\frac{\Gamma(\sum\limits_{i=1}^{2N}\bar{\mu}_i+1)\left(\prod\limits_{i=1}^{N}\frac{1}{\beta_{2i}^{\mu_{i}}\beta_{2i-1}^{\mu_i}}\right)}{\Gamma\left(1+\sum\limits_{i=1}^{2N}\bar{\mu}_i\right){T}^{\sum\limits_{i=1}^{2N}\bar{\mu}_i}}\frac{a_2^{\sum\limits_{i=1}^{2N}\bar{\mu}_i+\frac{1}{2}}} {a_1^{\frac{1}{2}}\Gamma(\bar{\mu})}$.
In order to simplify \eqref{final1a1}, we use the following transformation formula \cite[P. 287]{exton1976multiple}
\begin{equation*}
\textstyle
{}_{(1)}^{(1)}E_D^{(N)}[a,b_1,\cdots, b_N; c, a ;x_1,\cdots, x_N]=\left[ \prod_{i=2}^{N}(1-x_i)^{-b_i} \right] \times
\end{equation*}
\begin{equation}
\textstyle
F_D^{(N)}\left(b_1,a-b_2-\cdot -b_N,b_2,\cdots, b_N; c ;x_1,\frac{x_1}{1-x_2},\cdots, \frac{x_1}{1-x_N}\right) \label{eq:cov01}
\end{equation}
Now, the outage probability can be written as
\begin{equation*}
\textstyle
O_p=1-\left(\frac{a_2}{a_1}\right)^{\frac{1}{2}}  \prod\limits_{i=1}^{N}\left(\frac{a_2}{T\beta_{2i}+a_2}\right)^{\mu_i}\left(\frac{a_2}{T\beta_{2i-1}+a_2}\right)^{\mu_i}
\end{equation*}
\begin{equation}
\textstyle
 F_D^{(2N+1)}\bigg[ \frac{1}{2},\bar{\mu}_1, \bar{\mu}_1, \cdots,  
 \bar{\mu}_N, \bar{\mu}_N,1;1;\frac{(a_1-a_2)T\beta_1}{(T\beta_1+a_2)a_1},1-\frac{a_2}{a_1},  \frac{(a_1-a_2)T\beta_2}{(T\beta_2+a_2)a_1},\cdots,\frac{(a_1-a_2)T\beta_{2N}}{(T\beta_{2N}+a_2)a_1}\bigg] \label{final1a2}.
\end{equation}
Note that the outage probability expression when user channel experience Hoyt fading  is  in terms of single Lauricella's function of the fourth kind.
\vspace{-0.3in}
\subsection{Rate}
In this subsection, we derive the rate expression for $\kappa$-$\mu$ shadowed faded signals with integer $\mu$ and $\kappa$-$\mu$ shadowed faded interference with arbitrary parameter. 
\begin{theorem}
The rate  when the SoI experiences $\kappa$-$\mu$ shadowed fading with integer values of $\mu$ and  interferers experience $\kappa$-$\mu$ shadowed fading  is given by.
\begin{equation*}
\textstyle
R_{\kappa}=\frac{1}{\Gamma\left(1+\sum\limits_{i=1}^{N}\mu_i\right)}\frac{ \theta^ m }{\Gamma(\mu)(\lambda)^m}\sum\limits_{p=0}^{P} \frac{(m)_p\Big(1-\frac{\theta}{\lambda}\Big)^p\Gamma(\sum\limits_{i=1}^{N}\mu_i+p+\mu)}{(\mu)_p p!}\sum\limits_{i_1\cdots i_{2N}=0}^{p+\mu-1}\frac{(1-p-\mu)_{i_1+\cdots+i_{2N}}(\mu_1-m_1)_{i_1}\cdots(m_N)_{i_{2N}}}{\Big(1+\sum\limits_{i=1}^{N}\mu_i\Big)_{i_1+\cdots+i_{2N}}i_1!\cdots i_{2N}!} 
\end{equation*} 
\begin{equation}
\textstyle
\times \left(\sum\limits_{j=1}^{2N}i_j+\sum\limits_{j=1}^{N}\mu_j\right)^{-1} F_D^{(2N)}\left[1,\mu_1-m_1, \cdots, m_{N};\sum\limits_{i=1}^{N}\mu_i+1;1-\frac{\theta_1}{\theta},\cdots, 1-\frac{\lambda_N}{\theta}\right].\label{capacity20}
\end{equation}
\end{theorem}
\begin{proof}
See Appendix \ref{rate} for the proof.
\end{proof}
This is the final expression for rate and it is in terms of  sum of Lauricella's function of the fourth kind. 
Note that the rate expression when user experiences $\kappa$-$\mu$ ($m\rightarrow \infty$) can be obtained by   using the fact that $\lim\limits_{m\rightarrow \infty} \left(\frac{m}{a+m} \right)^m=\exp(-a)$ and $\lim\limits_{m\rightarrow \infty}(m)_p  \left(\frac{a}{a+m} \right)^p =a^p$ and it is given by
\begin{equation*}
\textstyle
R_{\kappa}=\frac{1}{\Gamma\left(1+\sum\limits_{i=1}^{N}\mu_i\right)}\frac{ e^{-\mu\kappa}  }{\Gamma(\mu)}\sum\limits_{p=0}^{P} \frac{(m)_p\Big(\mu\kappa\Big)^p\Gamma(\sum\limits_{i=1}^{N}\mu_i+p+\mu)}{(\mu)_p p!}\sum\limits_{i_1\cdots i_{2N}=0}^{p+\mu-1}\frac{(1-p-\mu)_{i_1+\cdots+i_{2N}}(\mu_1-m_1)_{i_1}\cdots(m_N)_{i_{2N}}}{\Big(1+\sum\limits_{i=1}^{N}\mu_i\Big)_{i_1+\cdots+i_{2N}}i_1!\cdots i_{2N}!} \times
\end{equation*} 
\begin{equation}
\textstyle
\left(\sum\limits_{j=1}^{2N}i_j+\sum\limits_{j=1}^{N}\mu_j\right)^{-1} F_D^{(2N)}\left[1,\mu_1-m_1, \cdots, m_{N};\sum\limits_{i=1}^{N}\mu_i+1;1-\frac{\theta_1}{\theta},\cdots, 1-\frac{\lambda_N}{\theta}\right].\label{capacity21}
\end{equation}
Similarly, when SoI and interferers both experience $\eta$-$\mu$ fading ($\mu=2\bar{\mu}, \kappa=\frac{1-\eta}{2\eta}, m=\bar{\mu}$) can be obtained. 

Note that the outage probability and rate expressions are derived for a given user location and the given interferers location. However, the outage probability and rate  of a typical user can be obtained by averaging over the distance from the desired BS and interferers.  For example, the outage probability and rate of a typical user when the users are uniformly distributed is given by
\begin{equation}
\textstyle
O_p=\int\limits_{r=0}^{R}O_p(r)f_R(r) \text{d}r,  \mbox{ and } R_{\kappa,t}=\int\limits_{r=0}^{R}R_{\kappa}(r)f_R(r)\text{d}r \label{oneone}
\end{equation} 
where $O_p(r)$ and $R_{\kappa}(r)$ are the outage probability and rate at a given distance $r$, respectively. Here $f_R(r)$ is the pdf of $r$ (distance between the user and desired BS) and it is given as $f_R(r)=\frac{2r}{R^2}, r\leqslant R, \text{ and } f_R(r)=0, r>R$

\section{Numerical and simulation Results}
In this section, we plot outage probability and rate using derived expression given in \eqref{kappa_mu5}, \eqref{capacity20}. Note that for analytical results, we have considered the number of terms $P$ to be $50$. We then compare these analytical results with the results obtained using simulations. For the simulation, we have considered a two tier network ($N=18$) with hexagonal structure having centre to edge distance $R=1000$ m as shown in Fig. \ref{fig:hexagonal}.   For simulation, we generate $\kappa$-$\mu$ shadowed RVs corresponding to SoI and independent $\kappa$-$\mu$ shadowed RVs corresponding to $N$ interferers and then compute the SIR using \eqref{macro_sir}.  Then using the simulated SIR we find the outage probability and rate and these are  averaged over $10^5$ realizations.  In all the plots the solid lines correspond to the simulation results, while the markers correspond to analytical results.

 Fig. \ref{cov_fig} depicts the outage probability with respect to distance from the serving BS for different values of $\kappa$, $\mu$ and $\alpha$. Here target SIR $T$ is assumed to be $0$ dB. It can be observed that simulation results  match with analytical results. Fig. \ref{rate_fig} depicts the rate with respect to distance from the serving BS for different values of $\kappa$, $\mu$ and $m$.  Fig. \ref{fig:ffr_sfr100} shows the outage probability when both SoI and CCI experience $\kappa$-$\mu$ fading. Here also target SIR $T$ is assumed to be $0$ dB.   \textbf{It has been mentioned in \cite{7870713} that the $\kappa$-$\mu$ shadowed fading model can be used to approximate the $\kappa$-$\mu$ distribution with arbitrary precision, by simply choosing a sufficiently large value of $m$. In order to plot Fig. \ref{fig:ffr_sfr100}, we have chosen  a sufficiently large value of $m$ in the derived expression when SoI and interferers both experience $\kappa$-$\mu$ shadowed fading.}  It can be seen that simulation results match with the analytical results\footnote{Though the figures just show few sets of results, we have carried out extensive simulations and in all cases the simulation results match with the analytical results}. We have also provided confidence interval (CI) for outage probability for different values of the system parameters in Table \ref{tab:Eavg}. Interestingly, it can be seen that even the $99\%$ confidence interval is fairly narrow when the number of iterations is  $10^5$. The CI values are calculated using the method in \cite{Confidence_compute}.

\begin{table*}[ht]
\centering
\begin{tabular}{|c|c|c|c|c|}
\hline
Parameters & no. of iterations & confidence interval ($95\%$) & confidence interval ($99\%$) & analytical values \\ \hline
$\alpha=3.6$, $r$=$600$ m & $10^2$  & 0.633 to 0.652 & 0.63 to 0.655 & 0.6492 \\ \hline
$\alpha=3.6$, $r$=$600$ m & $10^3$ &  0.646 to 0.652  & 0.645 to 0.652  & 0.6492 \\ \hline
$\alpha=3.6$, $r$=$600$ m & $10^4$ &  0.649 to 0.651 & 0.648 to 0.651 & 0.6492 \\ \hline
$\alpha=3.6$, $r$=$600$ m & $10^5$ &  0.649 to 0.65 & 0.649 to 0.65 & 0.6492 \\ \hline
$\alpha=3$, $r$=$800$ m & $10^2$  & 0.138 to 0.152 & 0.135 to 0.155 & 0.1471 \\ \hline
$\alpha=3$, $r$=$800$ m & $10^3$ &  0.146 to 0.151  & 0.145 to 0.151  & 0.1471 \\ \hline
$\alpha=3$, $r$=$800$ m & $10^4$ &  0.147 to 0.148 & 0.146 to 0.148 & 0.1471\\ \hline
$\alpha=3$, $r$=$800$ m & $10^5$ &  0.147 to 0.147 & 0.147 to 0.148 & 0.1471 \\ \hline
$\alpha=4$, $r$=$500$ m & $10^2$  & 0.868 to 0.881 & 0.866 to 0.884 &0.8783 \\ \hline
$\alpha=4$, $r$=$500$ m & $10^3$ & 0.875 to 0.879 & 0.874 to 0.88  & 0.8783 \\ \hline
$\alpha=4$, $r$=$500$ m & $10^4$ & 0.878 to 0.879 & 0.877 to 0.879 & 0.8783\\ \hline
\end{tabular}
\captionof{table}{Confidence interval for different values of the parameter. Here $\bar{\gamma}=\bar{\gamma}_i=1 \forall i$, $T=3$ dB, $\kappa=1.5, \mu=1.2$ $\kappa_i=1, \mu_i=1,m=m_i=10 \forall i$ and the sample size is assumed to be 100.}
\label{tab:Eavg}
\vspace{-0.4in}
\end{table*}

Now, we see the impact of fading parameters on the outage probability and rate. Fig. \ref{cov_mu} shows the variation in outage probability with respect to the parameter $\mu$ for different values of the shadowing parameter $m$. Observe that as the shadowing parameter $m$ increases, outage probability decreases. The reason of such a behaviour is as follows: the parameter $m$ is related to shadowing, i.e., as $m$ increases shadowing effect decreases.  Recall that $\kappa$-$\mu$ is a special case of $\kappa$-$\mu$ shadowed with $m\rightarrow \infty$ and therefore as the shadowing parameter increases, the outage probability  approaches the outage probability of $\kappa$-$\mu$ fading. It can  also be  observed that as $\mu$ increases, outage probability decreases for every value of $m$.  The reason for such a behavior is as follows:  as the parameter $\mu$ increases, the number of cluster increases and hence the outage probability decreases.

Fig. \ref{cov_m} shows the variation in outage probability with respect to shadowing parameter $m$ for different values of $\kappa$. Observe that there is no change in outage probability when $\kappa=0$ and  variation in outage probability increases as $\kappa$ increases. For example, when $\kappa=1$, the outage probability varies between $0.09389$ and $0.04842$ whereas when $\kappa=3$, the outage probability varies between $0.1784$ and $0.03143$. It can also be observed that when the shadowing parameter $m$ is small, the outage probability increases as the parameter $\kappa$ increases. However, at higher value of $m$, the outage probability decreases as  the parameter $\kappa$ increases.  

 Fig. \ref{rate_fig_m} depicts the variation in rate with respect to shadowing parameter $m$ for different values of $\kappa$. Similar to the outage probability, it can be observed that there is no change in rate when $\kappa=0$ and the variation in rate increases as $\kappa$ increases. For example, when $\kappa=1$, the rate varies between $1.553$ nats/Hz  and $1.693$ nats/Hz whereas  when $\kappa=3$, the rate varies between  $1.428$ nats/Hz and $1.716$ nats/Hz. It can be also observed that  when the shadowing parameter $m$ is small, the rate decreases as the parameter $\kappa$ increases. However, at higher value of $m$, the rate increases as  the parameter $\kappa$ increases.

\begin{figure}[ht]
 \centering
\includegraphics[scale=0.23]{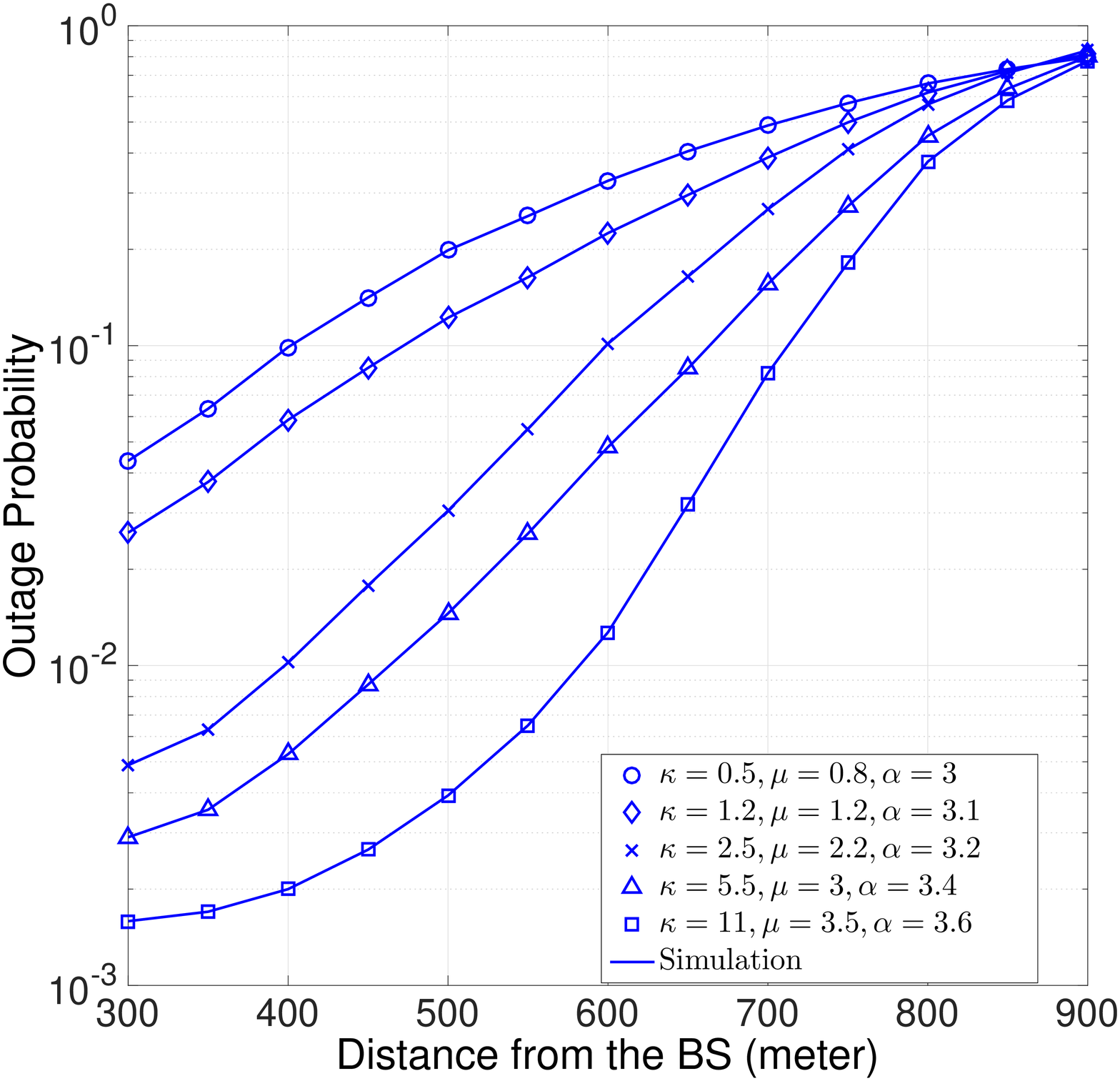}
\caption{Comparison of analytical results with simulation results of outage probability. Here $\bar{\gamma}=\bar{\gamma}_i=1 \forall i$, $T=0$ dB, $\kappa_i=1, \mu_i=1 \forall i\in \{1,\cdots,6\}$, $\kappa_i=0.8, \mu_i=0.8 \forall i\in \{7,\cdots,12\}$, $\kappa_i=0.5, \mu_i=0.5 \forall i\in \{13,\cdots,18\}$, $ m=8, m_i=10 \forall i$. }
\label{cov_fig}
\end{figure}

\begin{figure}[ht]
 \centering
\includegraphics[scale=0.22]{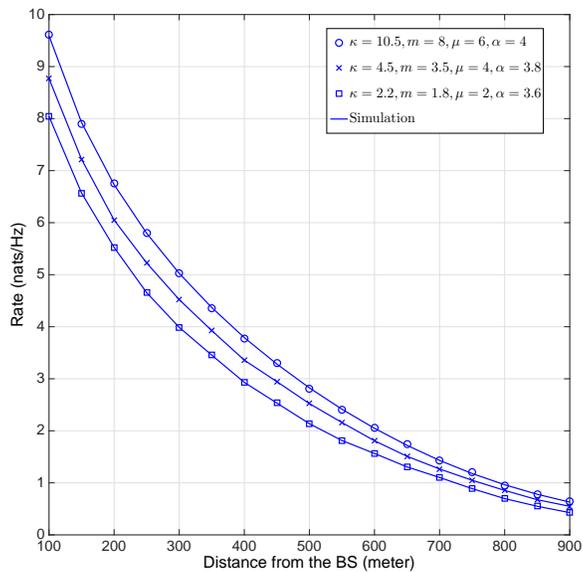}
\caption{Rate with respect to distance from  the BS. Here $\bar{\gamma}=\bar{\gamma}_i=1 \forall i$,  $\kappa_i=1.5, \mu_i=1.5 \forall i\in \{1,\cdots,6\}$, $\kappa_i=1, \mu_i=1 \forall i\in \{7,\cdots,12\}$, $\kappa_i=0.5, \mu_i=0.5 \forall i\in \{13,\cdots,18\}$, $m_i=10 \forall i$. }
\label{rate_fig}
\end{figure}

\begin{figure}[ht]
 \centering
\includegraphics[scale=0.22]{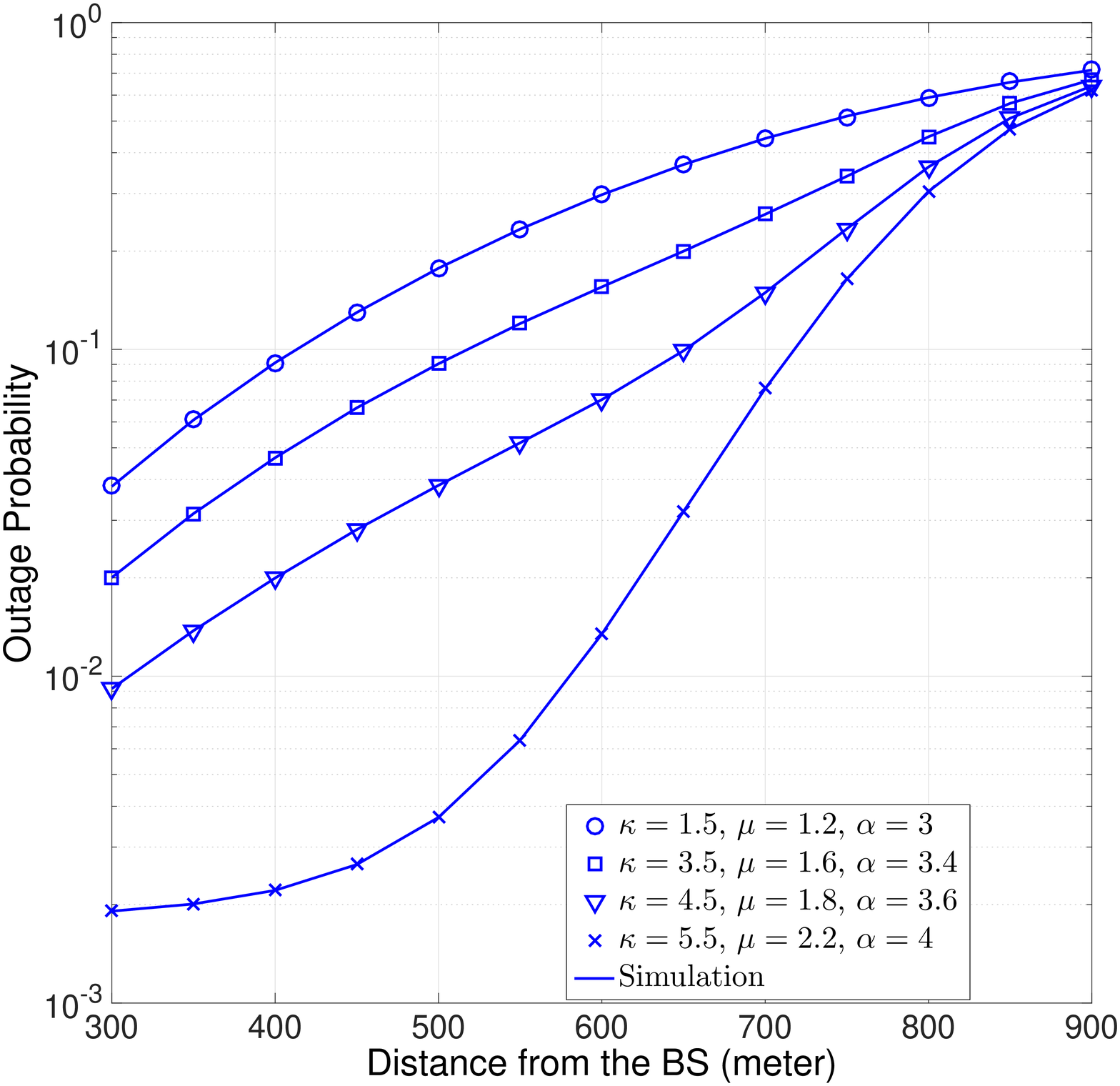}
\caption{Comparison of analytical results with simulation results of outage probability when both SoI and CCI experience $\kappa$-$\mu$ fading. Here $\bar{\gamma}=\bar{\gamma}_i=1 \forall i$, $T=0$ dB, $\kappa_i=1.2, \mu_i=1 \forall i\in \{1,\cdots,6\}$, $\kappa_i=1, \mu_i=0.8 \forall i\in \{7,\cdots,12\}$, $\kappa_i=0.8, \mu_i=0.5 \forall i\in \{13,\cdots,18\}$.}
\label{fig:ffr_sfr100}
\vspace{-0.2in}
\end{figure}

\begin{figure}[ht]
 \centering
\includegraphics[scale=0.23]{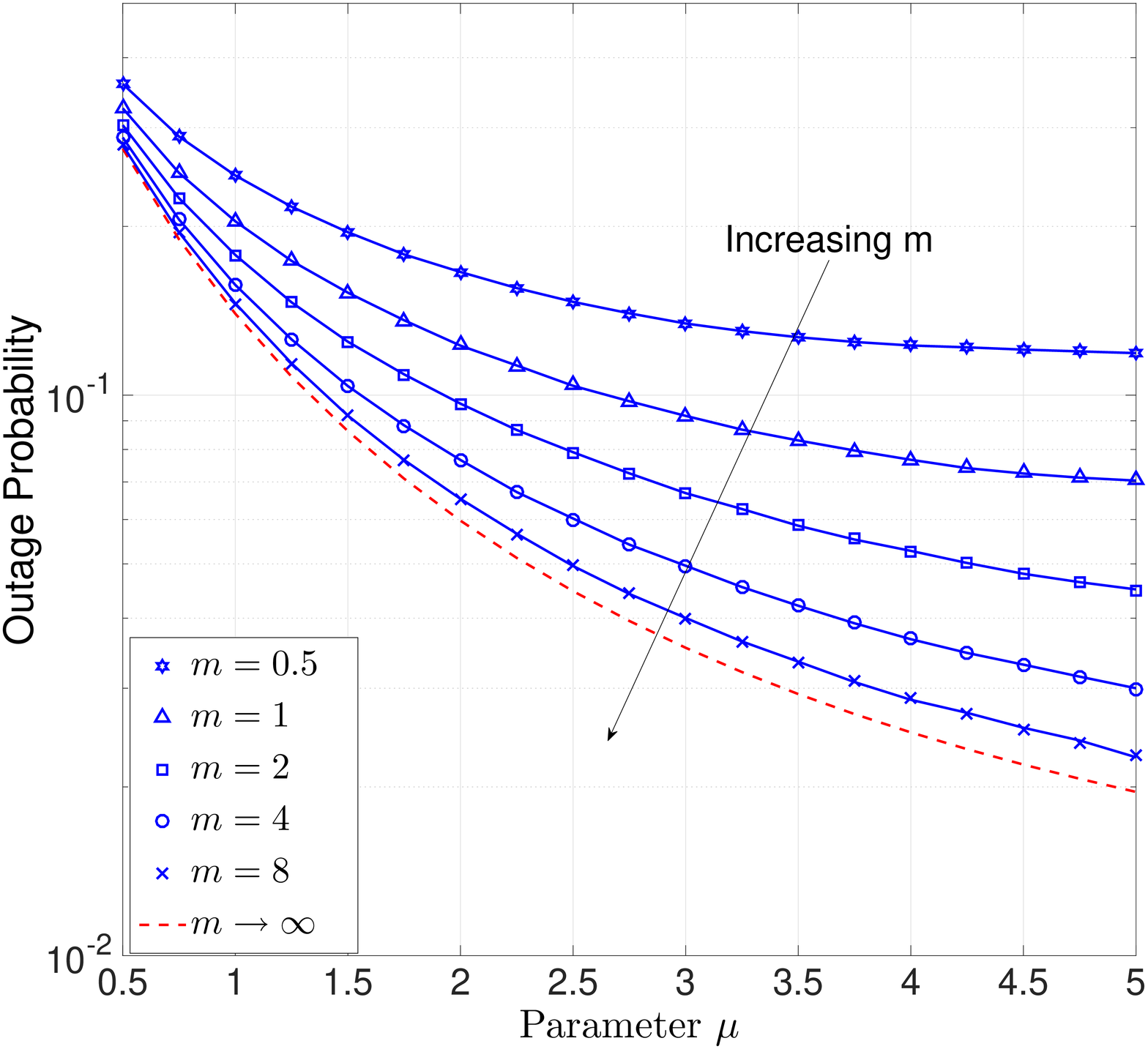}
\caption{Variation in outage probability with respect to  parameter $\mu$. Here $r=650$m $\bar{\gamma}=\bar{\gamma}_i=1 \forall i$, $T=0$dB, $\kappa=2$ $\kappa_i=1, \mu_i=1,m_i=10 \forall i$, and $\alpha=4$.}
\label{cov_mu}\vspace{-0.2in}
\end{figure}

\begin{figure}[ht]
 \centering
\includegraphics[scale=0.23]{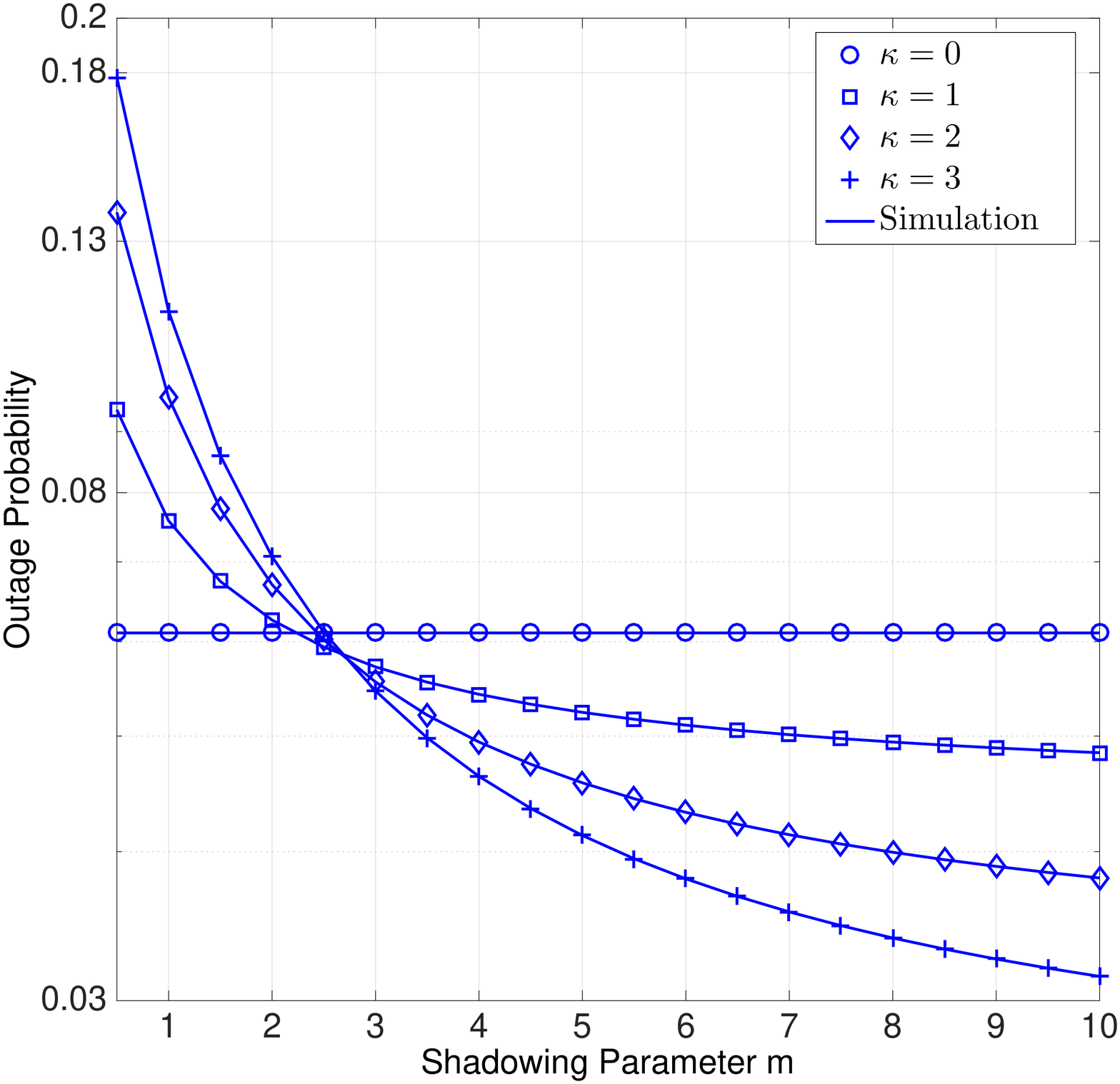}
\caption{Variation in outage probability with respect to shadowing parameter $m$. Here $r=650$m $\bar{\gamma}=\bar{\gamma}_i=1 \forall i$, $T=0$dB, $\mu=3$ $\kappa_i=1, \mu_i=1,m_i=10 \forall i$, and $\alpha=4$.}
\label{cov_m}
\vspace{-0.2in}
\end{figure}

\begin{figure}[ht]
 \centering
\includegraphics[scale=0.42]{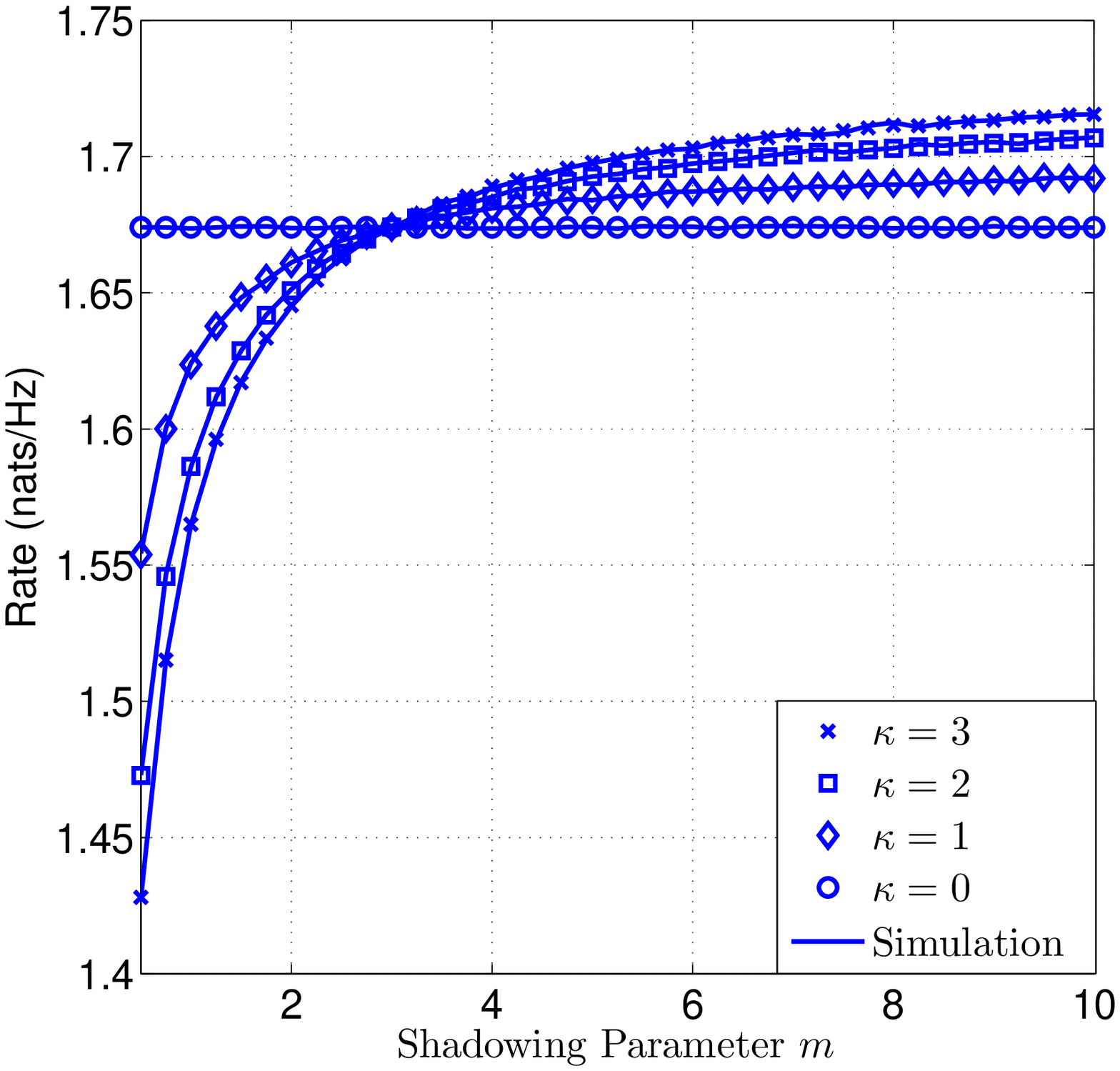}
\caption{Variation in rate with respect to shadowing parameter $m$. Here $r=650$m $\bar{\gamma}=\bar{\gamma}_i=1 \forall i$, $\mu=3$ $\kappa_i=0.5, \mu_i=1,m_i=10 \forall i$, and $\alpha=4$.}
\label{rate_fig_m}
\vspace{-0.2in}
\end{figure}
\vspace{-1.4in}
\subsection{FFR and SFR}
In this subsection, we compare the rate of FFR and SFR in the presence of $\kappa$-$\mu$ shadowed fading. The rate of FFR is given by \cite{7182316, 6047548}
\begin{equation}
\textstyle
 R_f = \int\limits_{r=0}^{R} R_{f,ce}(r)P[SIR>S_t]+\frac{1}{3}R_{f,ed}(r)P[SIR<S_t]\text{d}r \label{ffr}
\end{equation}
Here, the term $R_{f,ce}$ and $R_{f,ed}$ are the rate of cell-centre users and cell-edge users, respectively, and can be obtained using Theorem $2$. While computing $R_{f,ed}$, the number of interferers is taken to be $\frac{1}{3}$ of the number of interferers used for computing $R_{f,ce}$, since in cell-edge reuse $\frac{1}{3}$ frequency planning is used.  Similarly, $P[SIR>S_t]$ is the probability that users are classified as cell-centre users and can be obtained by using Theorem $1$. Note that this the $SIR$ for a reuse-1 system. Similarly, the rate of SFR is given by \cite{6047548}
\begin{equation}
 R_s = \int\limits_{r=0}^{R} R_{s,ce}(r)P[SIR>S_t]+R_{s,ed}(r)P[SIR<S_t]\text{d}r 
\label{sfr}
\end{equation}
Similar to the FFR scenario, all the terms given in \eqref{sfr} can be obtained using Theorem $1$ and Theorem $2$. However, in addition power control factor $\beta$ is used appropriately as in SFR $\beta$ is used for boosting the power of the cell-edge users.  For simulation, two tier network ($N=18$) with hexagonal structure has been considered. In each cell, $50$ physical resource blocks (PRBs) and  25 users are considered. The users are uniformly distributed in a cell and all resource blocks are uniformly shared among users. Further, we generate $\kappa$-$\mu$ shadowed RVs corresponding to SoI and $N$ interferers, and SIR per user and per PRB is evaluated. Similar to the \cite{7182316}, users with SIR higher than $S_t$ over $25$ PRBs or more than $25$ PRBs are classified as cell-centre users, otherwise they are classified as cell-edge users. For the FFR analytical computation, \eqref{ffr}, \eqref{final1} and \eqref{capacity20} are used, whereas for the SFR analytical computation, \eqref{sfr}, \eqref{final1} and \eqref{capacity20} are used.    Fig. \ref{fig:ffr_sfr} shows the rate of FFR and SFR system with respect to the shadowing parameter $m$. Firstly, it can be seen that as the shadowing parameter $m$ increases, the rate of both FFR and SFR schemes increases since shadowing effect reduces with increasing $m$. Secondly, it can be seen that the rate of FFR is always higher than the rate of SFR, even though SFR is more bandwidth efficient in a fully loaded system. Note that the similar results was presented in \cite{7452419} for Rayleigh fading and in \cite{6898829} for Nakagami-m fading. Now we have generalized it to $\kappa$-$\mu$ shadowed fading model which encompass both $\kappa$-$\mu$ and $\eta$-$\mu$ fading.  The intuitive reason for higher rate for FFR as follows: Note that the cell-edge users' SIR are quite low in a fully loaded system because of high interference. Even after employing SFR in the network, the cell-edge users' SIR are not significantly high. Whereas when FFR is employed, the cell-edge users' SIR increase significantly even in a fully loaded system as frequency reuse $\frac{1}{3}$ is employed at the cell-edge. Hence FFR performs better than SFR in a fully loaded system.

\begin{figure}[ht]
 \centering
\includegraphics[scale=0.24]{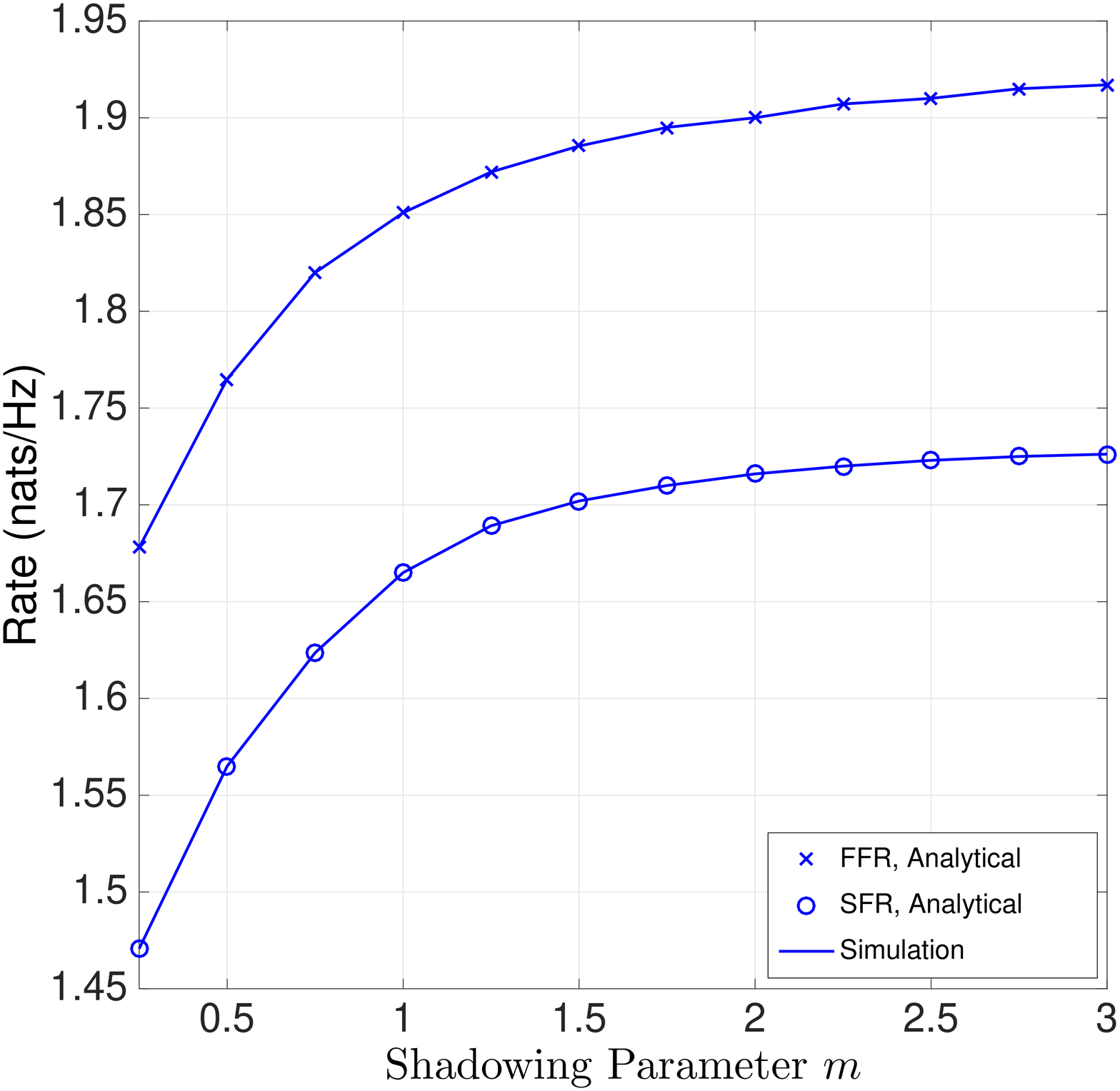}
\caption{Comparison of rate of FFR and SFR schemes. Here  $\bar{\gamma}=\bar{\gamma}_i=1 \forall i$,  $\mu=3$, $\kappa=2.5$, $\beta=2$ $\kappa_i=1, \mu_i=1.2, m_i=1.5 \forall i$, and $\alpha=3.4$.}
\label{fig:ffr_sfr}
\vspace{-0.2in}
\end{figure}

\vspace{-0.2in}

\section{Conclusion}
In this work, outage probability expression was derived when SoI and interferers both experience $\kappa$-$\mu$ shadowed fading. \textbf{The derived expression is valid for arbitrary SoI parameters, arbitrary $\kappa$ and $\mu$ parameters for all interferers and any value of the parameter $m$ for the interferers excepting the limiting value of $m\rightarrow \infty$. } The expression was given in terms of Pochhammer integral  where the integrands  only contain elementary functions. Also the expression was given in terms of a sum of Lauricella's function of the fourth kind, which can be easily evaluated numerically. Then the expression is simplified for the cases  when SoI experiences $\kappa$-$\mu$ fading and when both SoI and interferers experience $\eta$-$\mu$ fading.  Also, the outage probability is simplified for the case when SoI experience Hoyt fading and the interferers experience $\eta$-$\mu$ fading, and it was given in terms of single Lauricella's function of the fourth kind.   Further, the rate expression is derived for the case when  SoI experience  $\kappa$-$\mu$ shadowed fading with integer $\mu$ and interferers experience $\kappa$-$\mu$ shadowed fading with arbitrary parameters. Rate expression was also given in terms of  sum of Lauricella's function of the fourth kind.  We have also compared FFR and SFR  and shown that the FFR  outperforms SFR  in the presence of $\kappa$-$\mu$ shadowed fading. Finally, extensive simulation results were  given and these match with our analytical results. 
\begin{appendices}
\vspace{-0.2in}
\section{}
\label{outage_prob}
The outage probability expression can be written as $P(\text{SIR}<T)=P\bigg(I>\frac{g}{T}\bigg)$, where $T$ is the target SIR. The cdf of $I$ is given as \cite{6594884}
\begin{equation*}
\textstyle
F_I(y)=\left(\prod\limits_{i=1}^{N}\frac{1}{\theta_i^{\mu_i-m_i}\lambda_i^{m_i}}\right)\frac{y^{\sum\limits_{i=1}^{N}\mu_i}}{\Gamma\left(1+\sum\limits_{i=1}^{N}\mu_i\right)}\times
\end{equation*} 
\begin{equation}
\textstyle
\Phi_2^{(2N)}\bigg(\mu_1-m_1,\cdots, \mu_N-m_N,m_1\cdots,m_N;1+\sum\limits_{i=1}^{N}\mu_i;-\frac{y}{\theta_1},
\cdots,-\frac{y}{\theta_N}, -\frac{y}{\lambda_1},\cdots, \frac{y}{\lambda_N}\bigg)
\end{equation}
here $\Phi_2^{(2N)}(.)$ denotes confluent multivariate hypergeometric function and $\theta_i=\frac{\bar{\gamma}_i}{\mu_i(1+\kappa_i)}$, $\lambda_i=\frac{(\mu_i\kappa_i+m_i)\bar{\gamma}_i}{\mu_i(1+\kappa_i)m_i}$.
Using the cdf of $I$, the outage probability is given by
\begin{equation}
O_p=1-E_g\left(F_I\left(\frac{g}{T}\right)\right)
\end{equation}
Since $ g\sim S_{\kappa\mu}(\kappa,\mu,m, \bar{\gamma})$, outage probability can be simplified as 
\begin{equation*}
\textstyle
O_p=1-\int\limits_{0}^{\infty}\frac{g^{\mu-1} e^{-\frac{g}{\theta}}} {\theta^{\mu-m}\lambda^m\Gamma(\mu)}  {}_1 F_1 \left(m,\mu,\frac{g}{\theta}-\frac{g}{\lambda}\right) \left(\prod\limits_{i=1}^{N}\frac{1}{\theta_i^{\mu_i-m_i}\lambda_i^{m_i}}\right)\frac{(\frac{g}{T})^{\sum\limits_{i=1}^{N}\mu_i}}{\Gamma\left(1+\sum\limits_{i=1}^{N}\mu_i\right)} \times
\end{equation*} 
\begin{equation}
\textstyle \Phi_2^{(2N)}\bigg(\mu_1-m_1,\cdots, \mu_N-m_N,m_1\cdots,m_N;
1+\sum\limits_{i=1}^{N}\mu_i;-\frac{g}{T\theta_1},\cdots,-\frac{g}{T\theta_N}, -\frac{g}{T\lambda_1},\cdots, \frac{g}{T\lambda_N}\bigg) \text{d}g  \label{interchange}.
\end{equation}
Using the transformation of variables with $\frac{g}{\theta}=t$, above expression can be rewritten as
\begin{equation*}
\textstyle
O_p=1-K'\int\limits_{0}^{\infty}t^{\sum\limits_{i=1}^{N}\mu_i+\mu-1} e^{-t}{}_1 F_1 \left(m,\mu,\big(1-\frac{\theta}{\lambda}\big)t\right)\times
\end{equation*} 
\begin{equation}
\textstyle
\Phi_2^{(2N)}\bigg(\mu_1-m_1,\cdots, \mu_N-m_N,m_1\cdots,m_N;1+\sum\limits_{i=1}^{N}\mu_i;
-\frac{t\theta}{T\theta_1},\cdots,-\frac{t\theta}{T\theta_N}, -\frac{t\theta}{T\lambda_1},\cdots, \frac{t\theta}{T\lambda_N}\bigg) \text{d}t\label{eq:cov2}
\end{equation}
$K'=\frac{\left(\prod\limits_{i=1}^{N}\frac{1}{\theta_i^{\mu_i-m_i}\lambda_i^{m_i}}\right)}{\Gamma\left(1+\sum\limits_{i=1}^{N}\mu_i\right){T}^{\sum\limits_{i=1}^{N}\mu_i}}\frac{\theta^{\sum\limits_{i=1}^{N}\mu_i+m}} {\lambda^m\Gamma(\mu)}$.
In order to  simplify \eqref{eq:cov2}, we use the following relationship between  confluent multivariate hypergeometric function and ${}_{(1)}^{(k)}E_D^{(N)}(.)$ \cite[P. 95] {exton1976multiple}.
\begin{equation*}
\Gamma(a){}_{(1)}^{(k)}E_D^{(N)}[a,b_1,\cdots,b_N;c,c';x_1,\cdots, x_N]=
\end{equation*}
\begin{equation}
\int\limits_{0}^{\infty}e^{-t}t^{a-1}\Phi_2^{(k)}[b_1,\cdots,b_k;c;x_1t,\cdots,x_kt] 
\Phi_2^{(N-k)}[b_{k+1},\cdots,b_N;c';x_{k+1}t,\cdots,x_Nt]\text{d}t\label{eq:relationship1}.
\end{equation} 
Applying \eqref{eq:relationship1} to evaluate \eqref{eq:cov2}, one obtains
\begin{equation*} 
 O_p=1-K
{}_{(1)}^{(1)}E_D^{(2N+1)}\bigg(\sum\limits_{i=1}^{N}\mu_i+\mu, m,\mu_1-m_1,\cdots,
\end{equation*}
\begin{equation}
\textstyle
\mu_N-m_N,m_1\cdots,m_N;\mu,1+\sum\limits_{i=1}^{N}\mu_i;1-\frac{\theta}{\lambda},-\frac{\theta}{T\theta_1},\cdots,-\frac{\theta}{T\theta_N}, -\frac{\theta}{T\lambda_1},\cdots, -\frac{\theta}{T\lambda_N}\bigg) \label{kappa_mu3},
\end{equation}
where $K=\frac{\Gamma(\sum\limits_{i=1}^{N}\mu_i+\mu)\left(\prod\limits_{i=1}^{N}\frac{1}{\theta_i^{\mu_i-m_i}\lambda_i^{m_i}}\right)}{\Gamma\left(1+\sum\limits_{i=1}^{N}\mu_i\right){T}^{\sum\limits_{i=1}^{N}\mu_i}}\frac{\theta^{\sum\limits_{i=1}^{N}\mu_i+m}} {\lambda^m\Gamma(\mu)}$. 
Note that ${}_{(1)}^{(1)}E_D^{(N)}(.)$ is closely related to Lauricella function of fourth kind $F_D^{(N)}(.) $\cite{exton1976multiple}. A series expression for ${}_{(1)}^{(1)}E_D^{(N)}$ is given by
\begin{equation} 
{}_{(1)}^{(1)}E_D^{(N)}[a,b_1,\cdots, b_N;c, c';x_1,\cdots, x_N]=\sum\limits_{i_1\cdots i_N=0}^{\infty}\frac{(a)_{i_1+\cdots+i_N}(b_1)_{i_1}\cdots(b_N)_{i_N}}{(c)_{i_1}(c')_{i_2+\cdots+i_N}}\frac{x_1^{i_1}}{i_1!}\cdots\frac{x_N^{i_N}}{i_N!},\label{eq:lauricella3}
\end{equation}
with  region of convergence: $|x_1|<r_1,\cdots, |x_N|<r_N$; $r_2=r_3=\cdots=r_N$, and $r_1+r_N=1$. In order to obtain a series expression for ${}_{(1)}^{(1)}E_D^{(N)} (.)$ which converges, we use the following property of the ${}_{(1)}^{(1)}E_D^{(N)}(.)$ \cite[P.123]{exton1976multiple}.
\begin{equation*}
{}_{(1)}^{(1)}E_D^{(N)}[a,b_1,\cdots, b_N;c, c';x_1,\cdots, x_N]=(1-x_N)^{-a}\times
\end{equation*}
\begin{eqnarray}
\textstyle
{}_{(1)}^{(1)}E_D^{(N)}[a,b_1,\cdots,b_{N-1},c'-b_2-\cdot- b_N;c, c';\frac{x_1}{1-x_2},
\frac{x_2}{x_2-1},\frac{x_2-x_{3}}{x_{2}-1}\cdots,\frac{x_2-x_{N}}{x_{2}-1}]
\end{eqnarray}
and rewrite   \eqref{kappa_mu3} as
\begin{equation*}
\textstyle
O_p=1-K\big(\frac{T\theta_1}{\theta +T\theta_1}\big)^{\sum\limits_{i=1}^{N}\mu_i+\mu}{}_{(1)}^{(1)}E_D^{(2N+1)}\bigg[\sum\limits_{i=1}^{N}\mu_i+\mu, m,\mu_1-m_1,\cdots,
\end{equation*}
\begin{equation}
\textstyle
 \mu_N-m_N,m_1\cdots,m_{N-1},1;\mu,1+\sum\limits_{i=1}^{N}\mu_i;\frac{(\lambda-\theta)T\theta_1}{(T\theta_1+\theta)\lambda}, \frac{\theta}{\theta +T\theta_1}, \frac{\theta\theta_2-\theta\theta_1}{\theta_{2}(\theta+T\theta_1)},\cdots,\frac{\theta\lambda_{N}-\theta\theta_1}{\lambda_{N}(\theta+T\theta_1)}\bigg] \label{final}.
\end{equation}
It is apparent that $\left|\frac{(\lambda-\theta)T\theta_1}{(T\theta_1+\theta)\lambda}\right|<\frac{T\theta_1}{(T\theta_1+\theta)}$, $\left|\frac{\theta\theta_2-\theta\theta_1}{\theta_{2}(\theta+T\theta_1)}\right|<\frac{\theta}{\theta+T\theta_1}$, and $\frac{T\theta_1}{(T\theta_1+\theta)}+\frac{\theta}{\theta+T\theta_1}=1$. Hence, the series expression given in \eqref{final} converges. Therefore ${}_{(1)}^{(1)}E_D^{(N)} (.)$ can be evaluated using series expression.
\vspace{-0.2in}
\section{}
\label{evaluation}
In this Appendix, we further simplified outage probability given in \eqref{final1} in two different ways so that it can be evaluated easily. We use the following property of the ${}_{(1)}^{(1)}E_D^{(N)}(.)$ \cite[P.287]{exton1976multiple}.
\begin{eqnarray}
&{}_{(1)}^{(1)}E_D^{(N)}[c+c'-1,b_1,\cdots, b_N;c, c';x_1,\cdots, x_N]=\nonumber\\ &(1-x_1)^{-b_1}\cdots (1-x_N)^{-b_N}C_N^{(1)}[b_1,\cdots,b_{N},;1-c, 1-c';\frac{x_1}{1-x_1},\cdots,\frac{x_N}{1-x_N}]\label{trans},
\end{eqnarray}
where $C_N^{(k)}(.)$ is a generalization of the Horn function \cite[p. 104]{exton1976multiple}.
From \eqref{kappa_mu3}, it is clear that $c=\mu,c'=\sum\limits_{i=1}^{N}\mu_i+1$ and $a=\sum\limits_{i=1}^{N}\mu_i+\mu=c+c'-1$. Therefore using  \eqref{trans}, outage probability given in \eqref{kappa_mu3} (in Appendix \ref{outage_prob}) can be simplified as
\begin{equation*}
\textstyle
O_p=1-K(\frac{\theta}{\lambda})^{-m}(1+\frac{\theta}{T\theta_1})^{m_1-\mu_1}\cdots(1+\frac{\theta}{T\lambda_1})^{-m_N}\times
\end{equation*}
\begin{equation}
\textstyle
C_N^{(1)}\bigg[ m,\mu_1-m_1,\cdots, \mu_N-m_N,m_1\cdots,m_{N-1},1;1-\mu,-\sum\limits_{i=1}^{N}\mu_i;\frac{\lambda}{\theta}-1,\frac{-\theta}{\theta+T\theta_1},\cdots, \frac{-\theta}{\theta+T\lambda_N}\bigg] \label{final3}.
\end{equation}
The  Pochhammer integral  of $C_N^{(1)}$ is given by \cite[p. 105]{exton1976multiple}
\begin{equation*}
\frac{(2\pi i)^2 C_N^{(1)}[b_1,\cdots,b_{N},;c, c';x_1,\cdots,x_N]}{\Gamma(1-c)\Gamma(c+c')\Gamma(1-c')}= 
\end{equation*}
\begin{equation}
\int (-u)^{a-1}(u-1)^{-c-c'}(1+\frac{x_1}{u})^{-b_1}(1+ux_2)^{-b_2}\cdots (1+ux_N)^{-b_N}\text{d}u.
\end{equation}
The path of integration is a Pochhammer contour that is a double loop slung around the points $0$ and $1$. This loop starts from a point $p$ (say) between $0$ and $1$, encircles $0$ and $1$ in the positive direction, and then encircles the same two points again in the negative direction  \cite{exton1976multiple}.
Note that the  integrands of above integral expression only contains elementary functions. Thus, the outage probability expression is given in terms of Pochhammer integral where integrands consist only elementary functions and hence it can be easily evaluated.

We will now derive an alternate expression for ${}_{(1)}^{(1)}E_D^{(N)}(.)$ in terms of infinite sums of Lauricella's function of the fourth kind $F_D^{(N-1)}(.)$.  Then analytically we show that the infinite series can be truncated  to a finite series  with the number of terms being chosen such that the truncation error is lower than $\epsilon$, where $\epsilon$ is in the order of  $10^{-5}$ or lower. Moreover, as the number of terms increases,  the truncation error further decreases. The Matlab code to evaluate Lauricella's function of the fourth kind $F_D^{(N)}(.)$ is readily available and it can be downloaded from \cite{fd}.
The series  expression of ${}_{(1)}^{(1)}E_D^{(N)}$ for the outage probability  is given by
\begin{equation} 
{}_{(1)}^{(1)}E_D^{(N)}[a,b_1,\cdots, b_N;c, c';x_1,\cdots, x_N]=\sum\limits_{i_1\cdots i_N=0}^{\infty}\frac{(a)_{i_1+\cdots+i_N}(b_1)_{i_1}\cdots(b_N)_{i_N}}{(c)_{i_1}(c')_{i_2+\cdots+i_N}}\frac{x_1^{i_1}}{i_1!}\cdots\frac{x_N^{i_N}}{i_N!},\label{eq:lauricella2}
\end{equation}
 Now, using the property of  that Pochhammer symbol that $(a)_{m+n}=(a+m)_{n}(a)_m$ (equivalently, $\Gamma(a+m)(a+m)_n=(a)_{m+n}\Gamma(a)$),  $\eqref{eq:lauricella2}$ can be rewritten as 
\begin{equation}
\textstyle
{}_{(1)}^{(1)}E_D^{(N)}(.)=\sum\limits_{i_1\cdots i_N=0}^{\infty}\frac{\Gamma(a+i_1)(a+i_1)_{i_2+\cdots+i_N}(b_1)_{i_1}\cdots(b_N)_{i_N}}{\Gamma(a)(c)_{i_1}(c')_{i_2+\cdots+i_N}}\frac{x_1^{i_1}}{i_1!}\cdots\frac{x_N^{i_N}}{i_N!},\label{eq:lauricella4}
\end{equation}
\begin{equation}
\textstyle
{}_{(1)}^{(1)}E_D^{(N)}(.)=\sum\limits_{i_1=0}^{\infty} \frac{\Gamma(a+i_1)(b_1)_{i_1}x_1^{i_1}}{\Gamma(a)(c)_{i_1}i_1!}
\sum\limits_{i_2\cdots i_N=0}^{\infty}\frac{(a+i_1)_{i_2+\cdots+i_N}(b_2)_{i_2}\cdots(b_N)_{i_N}}{(c')_{i_2+\cdots+i_N}}\frac{x_2^{i_2}}{i_2!}\cdots\frac{x_N^{i_N}}{i_N!},\label{eq:lauricella5}
\end{equation}
The $N-1$ fold series expression given in \eqref{eq:lauricella5} is equivalent to the series expression  of $F_D^{N-1}(.)$ \cite{aalo} and hence
\begin{equation}
\scriptstyle
{}_{(1)}^{(1)}E_D^{(N)}[a,b_1,\cdots, b_N;c, c';x_1,\cdots, x_N]= \sum\limits_{i_1=0}^{\infty} \frac{\Gamma(a+i_1)(b_1)_{i_1}x_1^{i_1}}{\Gamma(a)(c)_{i_1}i_1!}F_D^{(N-1)}[a+i_1,b_2,\cdots, b_N; c';x_2,\cdots, x_N] \label{ED}.
\end{equation}
Note that the convergence condition for Lauricella's function of the fourth kind $F_D(.)$ will be satisfied if the convergence condition for  ${}_{(1)}^{(1)}E_D^{(N)}$ is satisfied, since $x_2,\cdots, x_N$ are same for both and the convergence criteria of  ${}_{(1)}^{(1)}E_D^{(N)}$ is $|x_1|<r_1,\cdots, |x_N|<r_N$; $r_2=r_3=\cdots=r_N$, and $r_1+r_N=1$ while the convergence criteria for the series expression of  $F_D(.)$ is $\max\{|x_2|,$ $\cdots|x_N|\}<1$.
Using \eqref{ED} and the following transformation
\begin{equation*}
\textstyle
F_D^{(N)}[a,b_1,\cdots, b_N; c ;x_1,\cdots, x_N]=(1-x_1)^{c-a-b_N} \left[ \prod_{i=1}^{N-1}(1-x_i)^{-b_i} \right] \times
\end{equation*}
\begin{equation}
\textstyle
F_D^{(N)}\bigg(c-a,b_1,\cdots, b_{N-1},c-b_1-\cdot-b_N; c ;\frac{x_N-x_1}{x_N-1},\cdots, \frac{x_{N-1}-x_1}{x_{N-1}-1},x_N\bigg), \label{eq:cov5}
\end{equation}
Eq. \eqref{final} can be rewritten in terms of $F_D(.)$ and is given by 
\begin{equation*}
\textstyle
O_p=K_2\sum\limits_{p=0}^{\infty} \frac{(m)_p\Big(1-\frac{\theta}{\lambda}\Big)^p\Gamma(\sum\limits_{i=1}^{N}\mu_i+p+\mu)}{(\mu)_p p!} F_D^{(2N)}\bigg(1-p-\mu,\mu_1-m_1,\cdots, \mu_N-m_N,m_1\cdots,m_N;
\end{equation*} 
\begin{equation}
\textstyle
 1+\sum\limits_{i=1}^{N}\mu_i; \frac{\theta}{\theta+T\theta_1},\cdots,\frac{\theta}{\theta+T\theta_N}, \frac{\theta}{\theta+T\lambda_1},\cdots, \frac{\theta}{\theta+T\lambda_N}\bigg) \label{kappa_mu401},
\end{equation}
where $K_2=\frac{\left(\prod\limits_{i=1}^{N}\left(\frac{\theta}{\theta+T\theta_i}\right)^{\mu_i-m_i}\left(\frac{\theta}{\theta+T\lambda_i}\right)^{m_i}\right)}{\Gamma\left(1+\sum\limits_{i=1}^{N}\mu_i\right)}\frac{\theta^ m }{\Gamma(\mu)(\lambda)^m}$. 

\vspace{-0.2in}
\section{}
\label{truncation}
 The outage probability with first $P$ terms is given by
\begin{equation*}
\textstyle
O_{p,P}=K_2\sum\limits_{p=0}^{P} \frac{(m)_p\Big(1-\frac{\theta}{\lambda}\Big)^p\Gamma(\sum\limits_{i=1}^{N}\mu_i+p+\mu)}{(\mu)_p p!} \times
\end{equation*} 
\begin{equation}
\textstyle
F_D^{(2N)}\bigg(1-p-\mu,\mu_1-m_1,\cdots, \mu_N-m_N,m_1\cdots,m_N; 1+\sum\limits_{i=1}^{N}\mu_i;\frac{\theta}{\theta+T\theta_1},\cdots,\frac{\theta}{\theta+T\theta_N}, \frac{\theta}{\theta+T\lambda_1},\cdots, \frac{\theta}{\theta+T\lambda_N}\bigg)\label{kappa_mu1}.
\end{equation}

In this appendix, we will bound the difference $e_P$ given by $O_{p}-O_{p,P}$ and show that to obtain $e_P< \epsilon$ (where $\epsilon$ is in the order of $10^{-5}$ or lower), the number of terms required, i.e., $P$ is finite. Also, as the number of terms $P$ increases, the $\epsilon$ decreases. $e_P=O_{p}-O_{p,P}$ is given by
\begin{equation*}
\textstyle
e_{P}=K_2\sum\limits_{p=P+1}^{\infty} \frac{(m)_p\Big(1-\frac{\theta}{\lambda}\Big)^p\Gamma(\sum\limits_{i=1}^{N}\mu_i+p+\mu)}{(\mu)_p p!} \times
\end{equation*} 
\begin{equation}
\textstyle
F_D^{(2N)}\bigg(1-p-\mu,\mu_1-m_1,\cdots, \mu_N-m_N,m_1\cdots,m_N;1+\sum\limits_{i=1}^{N}\mu_i;\frac{\theta}{\theta+T\theta_1},\cdots,\frac{\theta}{\theta+T\theta_N}, \frac{\theta}{\theta+T\lambda_1},\cdots, \frac{\theta}{\theta+T\lambda_N}\bigg) \label{kappa_mu2}.
\end{equation}
In order to bound $e_P$, first we bound $F_D^{(2N)}$. Note that 
\begin{equation}
F_D^{(N)}[a,b_1,\cdots, b_N; c ;x_1,\cdots, x_N]\leq F_D^{(N)}[a,b_1,\cdots, b_N; c ;x,\cdots, x]={}_2F_1[a,\sum\limits_{i=1}^{N}b_i,c,x].\label{in}
\end{equation}
Here $x \in \{x_1,\cdots, x_N\}$ is chosen such that above inequality holds. For example if  $a,b_1,\cdots, b_N$ are all positive, then $x=\max\{x_1,\cdots, x_N\}$. Applying inequality given in \eqref{in}, one can bound $e_P$ and it is given by
\begin{equation}
\textstyle
e_{P}\leq K_2 \sum\limits_{p=P+1}^{\infty} \frac{(m)_p\Big(1-\frac{\theta}{\lambda}\Big)^p\Gamma(\sum\limits_{i=1}^{N}\mu_i+p+\mu)}{(\mu)_p p!}  {}_2F_1[1-p-\mu, \sum\limits_{i=1}^{N}\mu_i,1+\sum\limits_{i=1}^{N}\mu_i,\frac{\theta}{\theta+Tx}]\label{error_f1}
\end{equation}
here $x \in \{\theta_1,\cdots, \theta_N,\lambda_1,\cdots, \lambda_N\}$.  We now consider two different cases\footnote{Note that $\sum\limits_{i=1}^{N}\mu_i$ represents the number of multipath cluster across all interferers and typically this is  one or greater than one.}: case $(i)$ $\sum\limits_{i=1}^{N}\mu_i\approx 1$ and case $(ii)$ $\sum\limits_{i=1}^{N}\mu_i>>1$. For case $(i)$, we use the transformation ${}_2F_1[a,b,c,x]=(1-x)^{c-a-b}{}_2F_1[c-a,c-b,c,x]$. Therefore ${}_2F_1(.)$ given in \eqref{error_f1} can be rewritten as
\begin{equation}
\textstyle
{}_2F_1[1-p-\mu, \sum\limits_{i=1}^{N}\mu_i,1+\sum\limits_{i=1}^{N}\mu_i,\frac{\theta}{\theta+Tx}]= \left(\frac{Tx}{\theta+Tx}\right)^{\mu+p}\, {}_2F_1[\sum\limits_{i=1}^{N}\mu_i+p+\mu, 1,1+\sum\limits_{i=1}^{N}\mu_i,\frac{\theta}{\theta+Tx}]
\end{equation}
and this can be bounded as
\begin{equation*}
\textstyle
{}_2F_1[\sum\limits_{i=1}^{N}\mu_i+p+\mu, 1,1+\sum\limits_{i=1}^{N}\mu_i,\frac{\theta}{\theta+Tx}]=\sum\limits_{i=0} ^{\infty}\frac{(\sum\limits_{i=1}^{N}\mu_i+p+\mu)_i(1)_i}{(1+\sum\limits_{i=1}^{N}\mu_i)_i}\frac{(\frac{\theta}{\theta+Tx})^i}{i!} 
\end{equation*}
\begin{equation}
\textstyle
\leq \sum\limits_{i=0} ^{\infty}\frac{(\sum\limits_{i=1}^{N}\mu_i+p+\mu)_i(1)_i}{(1)_i}\frac{(\frac{\theta}{\theta+Tx})^i}{i!}  = {_2F_1} \left(\sum\limits_{i=1}^{N}\mu_i+p+\mu, 1,1,\frac{\theta}{\theta+Tx} \right)=\left(\frac{Tx}{\theta+Tx}\right)^{-\sum\limits_{i=1}^{N}\mu_i-p-\mu}
\label{cp_series2}.
\end{equation}
Hence for case (i), ${}_2F_1(.)$ can be bounded as 
\begin{equation}
\textstyle
{}_2F_1[1-p-\mu, \sum\limits_{i=1}^{N}\mu_i,1+\sum\limits_{i=1}^{N}\mu_i,\frac{\theta}{\theta+Tx}]\leq \left(\frac{Tx}{\theta+Tx}\right)^{-\sum\limits_{i=1}^{N}\mu_i} \label{bound_i}.
\end{equation}
Therefore, using \eqref{error_f1} and \eqref{bound_i}, $e_P$   for case $(i)$  can be written as
\begin{equation}
\textstyle
e_{P}\leq K_3 \sum\limits_{p=P+1}^{\infty} \frac{\Gamma(P+m)}{\Gamma(P+1)}\frac{\Gamma(\sum\limits_{i=1}^{N}\mu_i+P+\mu)}{\Gamma(\mu+P)}  \Big(1-\frac{\theta}{\lambda}\Big)^P  \left(\frac{Tx}{\theta+Tx}\right)^{-\sum\limits_{i=1}^{N}\mu_i}\label{error_i}
\end{equation} 
\begin{equation*}
\textstyle
=K_3\left(\frac{Tx}{\theta+Tx}\right)^{-\sum\limits_{i=1}^{N}\mu_i} \frac{\Gamma(P+m+1)}{\Gamma(P+2)}\frac{\Gamma(\sum\limits_{i=1}^{N}\mu_i+P+\mu+1)}{\Gamma(\mu+P+1)}  \Big(1-\frac{\theta}{\lambda}\Big)^{P+1} 
\end{equation*}
\begin{equation}
\textstyle
\, _3F_2(1,P+m+1,\sum\limits_{i=1}^{N}\mu_i+P+\mu+1;P+2,\mu+P+1;1-\frac{\theta}{\lambda})\label{final_i}
\end{equation}
where $K_3=K_2\Gamma(\mu)/\Gamma(m)$.
For case $(ii)$ $\sum\limits_{i=1}^{N}\mu_i>>1$, ${}_2F_1(.)$ can be approximated as
\begin{equation}
\textstyle
 {}_2F_1[1-P-\mu, \sum\limits_{i=1}^{N}\mu_i,1+\sum\limits_{i=1}^{N}\mu_i,\frac{\theta}{\theta+Tx}] \approx{}_2F_1[1-P-\mu, \sum\limits_{i=1}^{N}\mu_i,\sum\limits_{i=1}^{N}\mu_i,\frac{\theta}{\theta+Tx}]\label{approx}.
\end{equation}
Using the following transformation \cite[P-19]{exton1976multiple}, ${}_2F_1[a,b;c;x]=(1-x)^{-a} {}_2F_1[a,c-b;c;x/(x-1)],$
the function ${}_2F_1(.)$ given in \eqref{approx} can be simplified as
\begin{equation}
\textstyle
{}_2F_1[1-P-\mu, \sum\limits_{i=1}^{N}\mu_i,\sum\limits_{i=1}^{N}\mu_i,\frac{\theta}{\theta+Tx}]=\left(\frac{Tx}{\theta+Tx}\right)^{P+\mu-1} \label{approx1}.
\end{equation}
Now, using \eqref{error_f1}, \eqref{approx} and \eqref{approx1}, $e_P$   for case $(ii)$  can be written as
\begin{equation*}
\textstyle
e_{P}\leq K_3  \sum\limits_{p=P+1}^{\infty} \frac{\Gamma(P+m)}{\Gamma(P+1)}\frac{\Gamma(\sum\limits_{i=1}^{N}\mu_i+P+\mu)}{\Gamma(\mu+P)}  \Big(1-\frac{\theta}{\lambda}\Big)^P  \left(\frac{Tx}{\theta+Tx}\right)^{P+\mu-1}
\end{equation*}
\begin{equation*} 
\textstyle
=K_3 \frac{\Gamma(P+m+1)}{\Gamma(P+2)}\frac{\Gamma(\sum\limits_{i=1}^{N}\mu_i+P+\mu+1)}{\Gamma(\mu+P+1)}  \Big(1-\frac{\theta}{\lambda}\Big)^{P+1} \left(\frac{Tx}{\theta+Tx}\right)^{P+\mu} \times \, 
\end{equation*}
\begin{equation}
\textstyle
_3F_2\Big(1,P+m+1,\sum\limits_{i=1}^{N}\mu_i+P+\mu+1;P+2,\mu+P+1;(1-\frac{\theta}{\lambda})(\frac{Tx}{\theta+Tx})\Big) \label{error}.
\end{equation}
Using the following identity from \cite{tricomi1951}
\begin{equation}
\textstyle
\frac{\Gamma(n+a)}{\Gamma(n+b)}= n^{a-b}\left(1+\frac{(a-b)(a+b-1)}{2n}+O(|n|^{-2})\right),\label{identity}
\end{equation}
and putting $1-\frac{\theta}{\lambda}=\frac{\mu\kappa}{\mu\kappa+m}$ into \eqref{error}, one can rewrite the truncation error  as
\begin{equation*}
\textstyle
e_P\leq K_3 \left(\underset{\text{I term}}{P^{m-1}+\frac{P^{m-2}(m+2)(m-1)}{2}}\right) \bigg(\underset{\text{II term}}{P^{\sum\limits_{i=1}^{N}\mu_i}+\frac{P^{\sum\limits_{i=1}^{N}\mu_i-1}(\sum\limits_{i=1}^{N}\mu_i-1)(\sum\limits_{i=1}^{N}\mu_i+2\mu+1)}{2}}\bigg)\times
\end{equation*}
\begin{equation}
\textstyle
 \underset{\text{III term}}{\left(\frac{\mu\kappa}{\mu\kappa+m} \right)^{P+1}} \underset{\text{IV term}}{ \left(\frac{Tx}{\theta+Tx}\right)^{P+\mu} } \underset{\text{V term}}{ {}_3F_2\Big(1,P+m+1,\sum\limits_{i=1}^{N}\mu_i+P+\mu+1;P+2,\mu+P+1;(\frac{\mu\kappa}{\mu\kappa+m} )(\frac{Tx}{\theta+Tx})\Big)}\label{cp_series4}.
\end{equation}
We now analyze the  upper bound of $e_P$ given in \eqref{cp_series4}. Let us first analyze the vth term given in \eqref{cp_series4} which contains the function ${}_3F_2(.)$. The vth term can be upper bounded as
\begin{equation}
\scriptstyle
{}_3F_2\Big(1,P+m+1,\sum\limits_{i=1}^{N}\mu_i+P+\mu+1;P+2,\mu+P+1;(\frac{\mu\kappa}{\mu\kappa+m} )(\frac{Tx}{\theta+Tx})\Big) \leq{}_3F_2\Big(1, P+\tilde{m}+1,\tilde{n}+P+\mu+1;P+2,\mu+P+1;(\frac{\mu\kappa}{\mu\kappa+m} )(\frac{Tx}{\theta+Tx})\Big).\label{cp_series40} 
\end{equation}
Here $\tilde{m}$ and $\tilde{n}$ are the smallest integer greater than  $m$ and $\sum\limits_{i=1}^{N}\mu_i$, respectively.  Now, using the following reduction formula given in \cite[Eq. (2)]{shpot2015clausenian},  one can rewrite the the vth term given in the right side of \eqref{cp_series40} as 
\begin{equation}
\textstyle
\text{vth term}\leq\sum\limits_{j_1=0}^{\tilde{m}-1}\sum\limits_{j_2=0}^{\tilde{n}}\frac{(\sum\limits_{i=1}^{N}\mu_i+P+\mu+1)_{j_1}(1)_{j_1+j_2}}{(P+2)_{j_1}(\mu+P+1)_{j_1+j_2}} {}_1F_0\big(1+j_1+j_2, (\frac{\mu\kappa}{\mu\kappa+m} )(\frac{Tx}{\theta+Tx}) \big)
\end{equation}
Using the fact that $(a)_n=\frac{\Gamma(a+n)}{\Gamma(a)}$, and ${}_1F_0(a; ;z)=(1-z)^{-a}$ one can rewrite the above term as
\begin{equation}
\textstyle
\text{vth term}\leq \sum\limits_{j_1=0}^{\tilde{m}-1}\sum\limits_{j_2=0}^{\tilde{n}}\frac{\Gamma(\sum\limits_{i=1}^{N}\mu_i+P+\mu+1+j_1)\Gamma(P+2)}{\Gamma(\sum\limits_{i=1}^{N}\mu_i+P+\mu+1) \Gamma(P+2+j_1) } \frac{\Gamma (1+j_1+j_2) \Gamma(\mu+P+1)}{\Gamma(\mu+P+1+j_1+j_2)}   \Big(1- \frac{\mu\kappa}{\mu\kappa+m} \frac{Tx}{\theta+Tx} \Big)^{-1-j_1-j_2}
\end{equation}
Further, using the identity given in \eqref{identity} and at sufficiently higher value of $P$, the vth term can be approximated as, vth term $\approx$
\begin{equation}
\textstyle
\sum\limits_{j_1=0}^{\tilde{m}-1}\sum\limits_{j_2=0}^{\tilde{n}}\frac{(P)^{j_1} \Gamma (1+j_1+j_2) }{(P)^{j_1} (P)^{j_1+j_2} }  \Big(1- \frac{\mu\kappa}{\mu\kappa+m} \frac{Tx}{\theta+Tx} \Big)^{-1-j_1-j_2}
\end{equation}
It is obvious from the above expression that as $P$ increases the $v$th term decreases. Now we show that the multiplication of first four terms  given in \eqref{cp_series4} also decreases as $P$ increases. Note that at higher value of $P$, one can approximate the multiplication of first four terms given in \eqref{cp_series4} as
\begin{equation*}
\textstyle
e_{P,4}\approx K_3P^{\sum\limits_{i=1}^{N}\mu_i+m}{\left(\frac{\mu\kappa}{\mu\kappa+m} \right)^{P+1}} { \left(\frac{Tx}{\theta+Tx}\right)^{P+\mu} } 
\end{equation*}
Recall that our goal is to show that as $P$ increases the multiplication of  first four terms, i.e.,  $e_{P,4}$, decreases. Therefore, differentiating $e_{P,4}$ with  respect to $P$, one obtains
\begin{equation}
\textstyle
e'_{P,4}= K_3 b^{P+1}c^{\mu+P}P^a(\frac{a}{p}+\log{b}+\log{c}) \label{diff1}
\end{equation}
where $a=\sum\limits_{i=1}^{N}\mu_i+m$, $b=\frac{\mu\kappa}{\mu\kappa+m}<1$, $c=\frac{Tx}{\theta+Tx}<1$. It is clear from \eqref{diff1} that at sufficiently higher value of $P$, the $e'_{P,4}<0$. as $b$ and $c$ are less than $1$. Hence the $e_{P,4}$ decreases as $P$ increases. It is also apparent that as $b$ and $c$ are higher, i.e., more close to $1$, more number of terms are required to truncate. In other words,  as $\mu$ increases or $\kappa$ increases, more number of terms are required to truncate,  since  $\frac{\mu\kappa}{\mu\kappa+m}$ also increases.  A similar analysis can be done for the case $(i)$ for which the upper bound is given in \eqref{final_i}.  
\vspace{-0.1in}
\section{}
\vspace{-0.1in}
\label{rate}
In order to derive the rate expression, first we need to derive the outage probability expression when SoI experience $\kappa$-$\mu$ shadowed fading  with integer $\mu$. Using the series expression of $F_D^{(N)}(.)$ given by
\begin{equation} 
F_D^{(N)}[a,b_1,\cdots, b_N; c;x_1,\cdots, x_N]=\sum\limits_{i_1\cdots i_N=0}^{\infty}\frac{(a)_{i_1+\cdots+i_N}(b_1)_{i_1}\cdots(b_N)_{i_N}}{(c)_{i_1+\cdots+i_N}}\frac{x_1^{i_1}}{i_1!}\cdots\frac{x_N^{i_N}}{i_N!},\label{eq:lauricella1}
\end{equation}
the expression given in \eqref{kappa_mu5} can be rewritten as
\begin{equation}
\scriptstyle
O_{p,P}=K_2\sum\limits_{p=0}^{P} \frac{(m)_p\Big(1-\frac{\theta}{\lambda}\Big)^p\Gamma(\sum\limits_{i=1}^{N}\mu_i+p+\mu)}{(\mu)_p p!}\sum\limits_{i_1\cdots i_{2N}=0}^{\infty}\frac{(1-p-\mu)_{i_1+\cdots+i_{2N}}(\mu_1-m_1)_{i_1}\cdots(m_N)_{i_{2N}}}{\Big(1+\sum\limits_{i=1}^{N}\mu_i\Big)_{i_1+\cdots+i_{2N}}}\frac{\left(\frac{\theta}{\theta+T\theta_1}\right)^{i_1}}{i_1!}\cdots\frac{\left(\frac{\theta}{\theta+T\lambda_N}\right)^{i_{2N}}}{i_{2N}!}.
\end{equation} 
When $\mu$ for the SoI fading is restricted to integer, $1-p-\mu$ will be integer and  above series expression can  be simplified using  the following property of Pochhammer symbol \cite[P. 14]{srivastava1985multiple}
\begin{equation}
(-n)_k= \left\{
\begin{array}{rl}
&\frac{(-1)^k n!}{(n-k)!}, 0\leq k\leq n,\\
&0,  k>n,
\end{array} \right.\label{pochhamer}
\end{equation}
\begin{equation*}
\textstyle \text{ and hence
 }
O_{p,P}=K_2\sum\limits_{p=0}^{P} \frac{(m)_p\Big(1-\frac{\theta}{\lambda}\Big)^p\Gamma(\sum\limits_{i=1}^{N}\mu_i+p+\mu)}{(\mu)_p p!} \times
\end{equation*} 
\begin{equation}
\textstyle
\sum\limits_{i_1\cdots i_{2N}=0}^{p+\mu-1}\frac{(1-p-\mu)_{i_1+\cdots+i_{2N}}(\mu_1-m_1)_{i_1}\cdots(m_N)_{i_{2N}}}{\Big(1+\sum\limits_{i=1}^{N}\mu_i\Big)_{i_1+\cdots+i_{2N}}}
\frac{\left(\frac{\theta}{\theta+T\theta_1}\right)^{i_1}}{i_1!}\cdots\frac{\left(\frac{\theta}{\theta+T\lambda_N}\right)^{i_{2N}}}{i_{2N}!}\label{integer}.
\end{equation}
Now, the rate is given in \cite{tse2005fundamentals},
$R_{\kappa} = \mathbb{E}[\text{ln}(1+\text{SIR})]$. For a positive random variable (RV) $\mathbb{E}[X] = \int_{t>0} \mathbb{P} [ {\ln}(1+\text{SIR})>t ] \text{d}t$. Since logarithm is monotonically increasing in SIR, we have $R_{\kappa} =  \int_{0}^\infty  \mathbb{P} [ \text{SIR} > e^t-1 ] \text{d}t$. With the help of \eqref{integer}, the rate is given by
\begin{equation*}
\textstyle
R_{\kappa}=\frac{1}{\Gamma\left(1+\sum\limits_{i=1}^{N}\mu_i\right)}\frac{\theta^ m }{\Gamma(\mu)(\lambda)^m}\int\limits_{t=0}^{\infty}\sum\limits_{p=0}^{P} \frac{(m)_p\Big(1-\frac{\theta}{\lambda}\Big)^p\Gamma(\sum\limits_{i=1}^{N}\mu_i+p+\mu)}{(\mu)_p p!} \times
\end{equation*}
\begin{equation}
\textstyle
\sum\limits_{i_1\cdots i_{2N}=0}^{p+\mu-1}\frac{(1-p-\mu)_{i_1+\cdots+i_{2N}}(\mu_1-m_1)_{i_1}\cdots(m_N)_{i_{2N}}}{\Big(1+\sum\limits_{i=1}^{N}\mu_i\Big)_{i_1+\cdots+i_{2N}}}
\frac{\left(\frac{\theta}{\theta+(e^t-1)\theta_1}\right)^{i_1+\mu_1-m_1}}{i_1!}\cdots\frac{\left(\frac{\theta}{\theta+(e^t-1)\lambda_N}\right)^{i_{2N}+m_{2N}}}{i_{2N}!}\text{d}t\label{capacity}.
\end{equation}
Rewriting the above expression one obtains, 
\begin{equation*}
\textstyle
R_{\kappa}=\frac{1}{\Gamma\left(1+\sum\limits_{i=1}^{N}\mu_i\right)}\frac{\theta^ m }{\Gamma(\mu)(\lambda)^m}\sum\limits_{p=0}^{P} \frac{(m)_p\Big(1-\frac{\theta}{\lambda}\Big)^p\Gamma(\sum\limits_{i=1}^{N}\mu_i+p+\mu)}{(\mu)_p p!} \times
\end{equation*} 
\begin{equation}
\textstyle
\sum\limits_{i_1\cdots i_{2N}=0}^{p+\mu-1}\frac{(1-p-\mu)_{i_1+\cdots+i_{2N}}(\mu_1-m_1)_{i_1}\cdots(m_N)_{i_{2N}}}{\Big(1+\sum\limits_{i=1}^{N}\mu_i\Big)_{i_1+\cdots+i_{2N}}i_1!\cdots i_{2N}!}
\int\limits_{t=0}^{\infty}\left(\frac{\theta}{\theta+(e^t-1)\theta_1}\right)^{i_1+\mu_1-m_1}\cdots\left(\frac{\theta}{\theta+(e^t-1)\lambda_N}\right)^{i_{2N}+m_{2N}}\text{d}t\label{capacity1}.
\end{equation}
To simplify it further we need to evaluate $\eta_t=\int\limits_{t=0}^{\infty}\prod\limits_{i=1}^{2N}\left(\frac{1}{1+(e^t-1)\theta_i}\right)^{\eta_i}\text{d}t$. $\eta_t$ can be simplified by following the method given in \cite[Appendix C]{7130668}. The details of the simplification are also given here for convenience (i.e. \eqref{rate2} to \eqref{eta_t}).
Using transformation of variable $e^{-t}=z$, one obtains
\begin{equation}
\textstyle
\eta_t=\int\limits_{0}^{1}\prod\limits_{i=1}^{2N}\left(\frac{1}{1+\theta_i (z^{-1}-1)}\right)^{\eta_i}\frac{1}{z}\text{d}z\label{rate2}.
\end{equation}
Rearranging the integrand of \eqref{rate2}, one can rewrite \eqref{rate2} as,
\begin{equation}
\textstyle
\eta_t=\prod\limits_{i=1}^{2N}(\theta_i)^{-\eta_i}\int\limits_{0}^{1}z^{\sum\limits_{i=1}^{2N}\eta_i-1} \prod\limits_{i=1}^{2N}\left(1-\frac{\theta_i-1}{\theta_i}z\right)^{-\eta_i}\text{d}z\label{rate3}.
\end{equation}
In order to simplify \eqref{rate3}, we use the following integral expression
\begin{equation}
\textstyle
\int\limits_{0}^{1}u^{a-1}(1-u)^{c-a-1}(1-ux_1)^{-b_1}\cdots (1-ux_N)^{-b_N}\text{d}u=\frac{\Gamma(a)\Gamma(c-a)}{\Gamma(c)}F_D^{(N)}\left[a,b_1, \cdots, b_N;c;x_1,\cdots, x_N\right] \label{expression}
\end{equation}
$\text{ for }\mathfrak{R}(a)>0, \mathfrak{R}(c-a)>0,$  where $\mathfrak{R}(a)$ is real part of $a$. Comparing \eqref{rate3} and \eqref{expression}, it is apparent that  \eqref{rate3} is equivalent to \eqref{expression} with $b_i=\eta_i \text{ }\forall i$, $a={\sum\limits_{i=1}^{2N}\eta_i}$, $c={\sum\limits_{i=1}^{2N}\eta_i+1}$ and $x_i=\frac{\theta_i-1}{\theta_i}\text{ } \forall i$. Hence, $\eta_t$ can be expressed as  
\begin{equation}
\textstyle
\eta_t=\prod\limits_{i=1}^{N}(\theta_i)^{-\eta_i} \frac{\Gamma(\sum\limits_{i=1}^{2N}\eta_i)\Gamma(1)}{\Gamma(\sum\limits_{i=1}^{2N}\eta_i+1)} F_D^{(2N)}\left[\sum\limits_{i=1}^{2N}\eta_i,\eta_1, \cdots, \eta_{2N};\sum\limits_{i=1}^{2N}\eta_i+1;\frac{\theta_1-1}{\theta_1},\cdots, \frac{\theta_N-1}{\theta_N}\right].
\end{equation}
\begin{equation}
\textstyle \eta_t=\left(\frac{\prod\limits_{i=1}^{N}(\theta_i)^{-\eta_i}}{\sum\limits_{i=1}^{2N}\eta_i}\right) 
F_D^{(2N)}\left[\sum\limits_{i=1}^{2N}\eta_i,\eta_1, \cdots, \eta_{2N};\sum\limits_{i=1}^{2N}\eta_i+1;\frac{\theta_1-1}{\theta_1},\cdots, \frac{\theta_N-1}{\theta_N}\right]. \label{eta_t}
\end{equation}
Using the transformation of the Lauricella's function of the fourth kind \cite[p. 286]{exton1976multiple}.
\begin{equation*}
\textstyle
F_D^{(N)}[a,b_1,\cdots, b_N; c ;x_1,\cdots, x_N]=\left[ \prod\limits_{i=1}^{N}(1-x_i)^{-b_i} \right]F_D^{(N)}\left(c-a,b_1,\cdots, b_N; c ;\frac{x_1}{x_1-1},\cdots, \frac{x_N}{x_N-1}\right) 
\end{equation*}
 $\eta_t$ can be further simplified as
\begin{equation}
\textstyle 
\eta_t=\left(\sum\limits_{i=1}^{2N}\eta_i\right)^{-1} 
F_D^{(2N)}\left[1,\eta_1, \cdots, \eta_{2N};\sum\limits_{i=1}^{2N}\eta_i+1;1-\theta_1,\cdots, 1-\theta_N\right]\label{rate4}.
\end{equation}
Using above expression and \eqref{capacity1}, the rate expression is given by
\begin{equation*}
\textstyle
R_{\kappa}=\frac{1}{\Gamma\left(1+\sum\limits_{i=1}^{N}\mu_i\right)}\frac{ \theta^ m }{\Gamma(\mu)(\lambda)^m}\sum\limits_{p=0}^{P} \frac{(m)_p\Big(1-\frac{\theta}{\lambda}\Big)^p\Gamma(\sum\limits_{i=1}^{N}\mu_i+p+\mu)}{(\mu)_p p!} \sum\limits_{i_1\cdots i_{2N}=0}^{p+\mu-1}\frac{(1-p-\mu)_{i_1+\cdots+i_{2N}}(\mu_1-m_1)_{i_1}\cdots(m_N)_{i_{2N}}}{\Big(1+\sum\limits_{i=1}^{N}\mu_i\Big)_{i_1+\cdots+i_{2N}}i_1!\cdots i_{2N}!} \times
\end{equation*} 
\begin{equation}
\textstyle
\left(\sum\limits_{j=1}^{2N}i_j+\sum\limits_{j=1}^{N}\mu_j\right)^{-1} F_D^{(2N)}\bigg[1,\mu_1-m_1, \cdots, m_{N};
\sum\limits_{i=1}^{N}\mu_i+1; 1-\frac{\theta_1}{\theta},\cdots, 1-\frac{\lambda_N}{\theta}\bigg].\label{capacity2}
\end{equation}

\end{appendices}
\vspace{-0.2in}
\bibliographystyle{IEEEtran}
\bibliography{bibfile}
\end{document}